\documentclass[12pt,onecolumn,draftclsnofoot]{IEEEtran}
\IEEEoverridecommandlockouts
\usepackage{cite}
\usepackage{amsmath,amssymb,amsfonts}
\DeclareMathOperator*{\argmin}{arg\,min}

\usepackage{algorithmic}
\usepackage{graphicx}
\usepackage{epstopdf}
\usepackage{booktabs} 
\usepackage{epsfig}
\usepackage{textcomp}
\usepackage{float}
\usepackage{xcolor}
\usepackage{subfigure}
\usepackage{amsthm}
\usepackage{caption}
\usepackage{stfloats}
\usepackage{verbatim}
\usepackage{array}
\usepackage{bbm}
\newtheorem{Theorem}{Theorem}
\newtheorem{Lemma}{Lemma}
\newtheorem{Proposition}{Proposition}

\def\BibTeX{{\rm B\kern-.05em{\sc i\kern-.025em b}\kern-.08em
		T\kern-.1667em\lower.7ex\hbox{E}\kern-.125emX}}
\allowdisplaybreaks
\begin{document}
	\title{On User Association in Large-Scale Heterogeneous LEO Satellite Network\\
	}
	\author{Yuan Guo,~\IEEEmembership{Student Member,~IEEE,}
		Christodoulos~Skouroumounis,~\IEEEmembership{Member,~IEEE,}
		Symeon~Chatzinotas,~\IEEEmembership{Fellow,~IEEE,}
		and~Ioannis~Krikidis,~\IEEEmembership{Fellow,~IEEE}
			\thanks{This work received funding from the European Research Council (ERC) under the European Union’s Horizon 2020 research and innovation programme (Grant agreement No. 819819). This work was also supported by the  European Union HORIZON programme under iSEE-6G (Grant agreement No. 101139291), and by the Luxembourg National Research Fund (FNR) under the project MegaLEO (C20/IS/14767486). }
		\IEEEcompsocitemizethanks{\IEEEcompsocthanksitem Yuan Guo, Christodoulos Skouroumounis, and Ioannis Krikidis are with the IRIDA Research Centre for Communication Technologies,
			Department of Electrical and Computer Engineering, University of Cyprus,
			e-mail: \{yguo0001, cskour03, krikidis\}@ucy.ac.cy. Symeon Chatzinotas is with the Interdisciplinary Centre for Security, Reliability and Trust (SnT), University of Luxembourg, e-mail: symeon.chatzinotas@uni.lu. Part of this work was presented at the  IEEE International Conference on Communications, Rome, Italy, May 2023 \cite{Yuan}.}}
	\maketitle
	\vspace{-12mm}
	\begin{abstract}
			In this paper, we investigate the performance of large-scale heterogeneous low Earth orbit (LEO) satellite networks in the context of three association schemes. In contrast to existing studies, where single-tier LEO satellite-based network deployments are considered, the developed framework captures the heterogeneous nature of real-world satellite network deployments. More specifically, we propose an analytical framework to evaluate the performance of multi-tier LEO satellite-based networks, where the locations of LEO satellites are approximated as points of independent Poisson point processes, with different density, transmit power, and altitude. We propose three association schemes for the considered network topology based on: 1) the Euclidean distance, 2) the average received power, and 3) a random selection. By using stochastic geometry tools, analytical expressions for the association probability, the downlink coverage probability, as well as the spectral efficiency are derived for each association scheme, where the interference is considered. Moreover, we assess the achieved network performance under several different fading environments, including low, typical, and severe fading conditions, namely non-fading, shadowed-Rician and Rayleigh fading channels, respectively. Our results reveal the impact of fading channels on the coverage probability, and illustrate that the average power-based association scheme outperforms in terms of achieved coverage and spectral efficiency performance against the other two association policies. Furthermore, we highlight the impact of the proposed association schemes and the network topology on the optimal number of LEO satellites, providing guidance for the planning of multi-tier LEO satellite-based networks in order to enhance network performance.
	\end{abstract}
	\vspace{-3mm}
	\section{Introduction}
	The rapid evolution of massive Internet-of-Thing (IoT) devices demands the emergence of sixth-generation (6G) networks \cite{Nguyen}.  As we transition from fifth-generation (5G) to 6G, satellite communication networks are anticipated to play a pivotal role in realizing the vision of a fully connected and intelligent world. Satellite communications can significantly contribute to achieving the goals of 6G by providing seamless global coverage and enhanced connectivity in remote, rural, and underserved areas, as well as reinforcing the resilience of communication networks in disaster-stricken regions \cite{Fang,Yaacoub,Xiangming}. Moreover, satellite communication systems hold an immense potential for enabling advanced applications such as high-speed internet, ultra-low latency communication, and massive IoT deployments, empowering sectors like smart cities, agriculture, and transportation \cite{Centenaro}. Hence, the concept of satellite communication networks has emerged as a critical research paradigm, requiring a comprehensive understanding of its potential impact, challenges, and opportunities \cite{Qu,Di,Mohammed}.
	
	Satellite communication networks are characterized by a diverse range of orbital configurations, such as geostationary (GEO), medium Earth orbit (MEO), and low Earth orbit (LEO) satellites. Owing to the unique benefits of LEO satellites, such as minimal latency, superior spatial flexibility, and low-cost deployment, the employment of LEO satellites has emerged as an attractive solution for 6G networks \cite{Liu,Xie,Hraishawi}. However, the proliferation of LEO satellites and mega constellations pose significant challenges on network modelling, design and analysis, necessitating the employment of sophisticated mathematical tools that accurately capture the characteristics of large-scale LEO satellite networks. Recently, tools from the field of stochastic geometry have been leveraged to analyze the performance of large-scale LEO satellite-based communication systems, highlighting its effectiveness as a powerful and tractable mathematical tool for assessing the impact of key design parameters on network performance \cite{Wang,Akram1,Akram2,Jung,Jia}. Specifically, the authors in \cite{Wang} study the performance of large-scale LEO satellite communication networks, illustrating the impact of satellites' density and constellation altitude on the coverage and communication latency. Moreover, the authors  in \cite{Akram1} investigate the coverage performance of  LEO satellite networks, by assuming that the satellite spatial deployment follows a homogeneous Poisson point process (PPP), where a novel path-loss model is introduced for the satellite communication networks. This work is further extended in \cite{Akram2}, where the optimal satellite altitude for attaining the highest coverage probability is numerically investigated. Furthermore, the authors in \cite{Jung} study the downlink performance of LEO satellite-based networks by modelling the spatial deployments of LEO satellites according to a homogeneous binomial point process (BPP) and by investigating an iterative algorithm to maximize the transmission rate and system throughput. The authors in \cite{Jia} employ stochastic geometry to analyze uplink performance of large-scale LEO satellite networks, demonstrating the optimal LEO satellite density for achieving the maximum ergodic capacity.  Additionally, the authors in \cite{Okati2} adopt a non-homogeneous PPP to model a massive LEO network and demonstrate the optimum altitude of LEO satellites as well as number of orthogonal frequency channels for achieving the highest network throughput. Nevertheless, these studies highlight the capability of stochastic geometry in modelling large-scale LEO satellite-based networks, under the assumption that  LEO satellites are deployed at the same altitude, opposed to real-world LEO satellite-based configurations \cite{Xie,Wang,Talgat}. Recently, the multi-altitude setup of LEO satellites has been discussed in several works, by considering the same type of LEO satellites deployed at various altitudes  \cite{Talgat,Qiu10119066,choi2023cox,Wang10184196}. In specific, the authors in \cite{Talgat} and \cite{choi2023cox} model the multi-tier LEO satellites by using binomial point process and Cox point process, respectively. In addition, a simulation-based study is presented in \cite{Qiu10119066} to evaluate the  effect of interference on the downlink systems. Moreover, the authors in \cite{Wang10184196} investigate the performance of multi-hop routing in multi-tier LEO satellite networks. Although these works have explored the altitude variations in LEO satellite-based networks, the real-world multi-tier LEO satellites network exhibit more significant heterogeneity, such as different transmit power and beamwidth of LEO satellites, etc, which are overlooked from the current literature.
	
	Given the heterogeneity presented in large-scale LEO satellite-based networks, an appropriate association scheme, which methodically assigns each user to a specific serving satellite, is of paramount importance for enhancing the end-user performance. Various association schemes for conventional terrestrial networks have been extensively explored in the literature, that are either based on contact distance, cell area, received signal strength, or other quality of service metrics \cite{Marco,Jo,Alzenad8833522}. In specific, the authors in \cite{Marco} propose two association schemes for millimeter wave networks, by taking into consideration of the path-loss and the received signal power, respectively. Additionally, a flexible association scheme is proposed in \cite{Jo}, where a mobile user is associated with the strongest base stations (BSs) in terms of long-term averaged biased power. The authors in \cite{Alzenad8833522} study a two-tier vertical heterogeneous network consisting of terrestrial and aerial BSs, where an average received power-based association is adopted. However,  association schemes for large-scale heterogeneous LEO satellite networks are missing from the current literature; the majority of works that investigate massive LEO satellite networks either ignore the heterogeneity or assume simplistic distance-based association schemes \cite{Okati2,Akram1,Akram2,Jung,Jia}. This is mainly due to the challenges arising from the different altitudes of multi-tier LEO satellites that further escalate the heterogeneity of networks \cite{Yuan}. More specifically, unlike conventional terrestrial networks, where the omission of BSs’ heights does not deprive from the accuracy of  analysis, the altitude variation among multi-tier LEO satellites  serves as a critical factor, significantly affecting the coverage area and signal strength.  The diverse altitudes of LEO satellites not only lead to a higher level of complexity in the design of appropriate association schemes but also pose substantial challenges on analyzing the network performance.
	
	Motivated by the above discussion, the aim of this work is to fill these gaps by modelling and analyzing the performance of multi-tier LEO satellite-based networks. Specifically, we extend our previous work \cite{Yuan} by developing efficient association schemes to enhance the downlink performance of the considered large-scale LEO satellite networks. The main contributions of this paper are summarized as follows:
	\begin{itemize}
		\item We develop an analytical framework based on stochastic geometry, which sheds light on the modelling and analysis of large-scale heterogeneous LEO satellite networks from a macroscopic point-of-view. The developed framework takes into account the real-world heterogeneous nature of large-scale LEO satellite networks, which involves the deployment of a general number of tiers of satellites at various altitudes with different transmit power and density; the  locations of LEO satellites are modelled as the points of multiple independent random point processes.
		\item We propose three association schemes for the considered LEO satellite networks, each with a distinct focus that collectively enhances network performance. The first scheme aims at minimizing communication path-loss, focusing on reducing the attenuation of signals  to improve the quality of the communication link. The second scheme is dedicated to maximizing received signal strength, thereby ensuring robust communication by enhancing the power of received signals. Finally, the third scheme addresses the computational limitations of IoT devices by streamlining the association process, making satellite network access more feasible for devices with restricted computational capabilities.
		\item  We evaluate the performance of each association scheme, in terms of association probability, coverage probability and spectral efficiency, by considering different fading scenarios, namely, shadowed-Rician (SR), Rayleigh and non-fading channels. By leveraging tools from stochastic geometry, we initially evaluate the distance and the average signal power distributions, based on which analytical expressions for the association probability with the employment of each association schemes are derived. Additionally, we derive the analytical expressions for the achieved coverage probability and spectral efficiency in the context of proposed association schemes, with SR, Rayleigh and non-fading channels.
		\item Our results reveal the negative effect of fading channels on the coverage probability for low signal-to-interference-plus-noise ratio (SINR) thresholds, and their beneficial impact in high SINRs. In addition, the numerical results illustrate that the average power-based association scheme outperforms the distance-based association schemes in terms of coverage probability and spectral efficiency, while the random selection-based scheme acts as a performance lower bound. Finally, we highlight the trade-off imposed by the number of LEO satellites between enhancing the coverage performance and escalating the observed multi-user interference, while the optimal number of LEO satellites is demonstrated for achieving the highest coverage probability. 
	\end{itemize}
	The rest of the paper is organised as follows: Section \ref{sm} introduces the considered system model. Section \ref{associationscheme} presents our proposed association schemes along with the detailed analysis of association probability. Section \ref{downper} presents the downlink network performance. Finally, numerical   results are presented in Section \ref{numsim}, followed by our conclusions in Section \ref{conclu}.
	
	\textit{Notation}: $\Gamma(\cdot)$ and $\gamma(\cdot,\cdot)$ denote the Gamma and the lower incomplete Gamma functions, respectively;$ \ {}_1 F_1\left(\cdot;\cdot;\cdot\right)$  is the confluent hypergeometric function of the first kind;  $(n)!$ denotes the factorial of $n$; $(m)_z$ denotes the Pochhammer symbol; $\Im\{ \cdot \}$ denotes the imaginary operator; $\mathcal{H}(\cdot)$ denotes the Heaviside function.
	\vspace{-5mm}
	\section{System model}\label{sm}
	\begin{table*}[t]\centering
		\caption{Summary of notations.}\label{Table1}
		\scalebox{0.8}{
			\begin{tabular}{| l | l || l | l |}\hline
				\textbf{Notation} & \textbf{Description} & \textbf{Notation} & \textbf{Description}\\\hline
				$\Phi_k$, $\lambda_k$ & PPP of the $k$-th tier LEO satellites with intensity $\lambda_k$ & $H_k$ & Altitude of the $k$-th tier LEO satellites  \\\hline
				$x_{i,k}$ & Location of the $i$-th LEO satellite that belongs to the $k$-th tier &   $\alpha$ & Path-loss exponent \\\hline
				$r_{i,k}$ & Distance between the typical gUE and the LEO satellite located at $x_{i,k}$& $R_{\oplus}$ & Radius of the Earth  \\\hline
				$G_{m,k}^L,G_{s,k}^L$ & Main- and side-lobe antenna gain of the $k$-th tier LEO satellites& $G_m^U,G_s^U$ & Main- and side-lobe antenna gain of the gUEs \\\hline
				$R_{\mathrm{max}_k}$ &Maximum distance from a gUE to the $k$-th tier satellites& $K$ & The number of tiers of LEO satellites \\\hline
				$(b,\Omega,m)$ &  Parameters of the shadowed-Rician fading model&  $P_k$ & Transmit power of the $k$-th tier LEO satellites  \\\hline
				$N_k$ & Average number of the $k$-th tier LEO satellites & $\sigma^2$& Thermal noise power \\\hline
		\end{tabular}}
		\vspace{-8mm}
	\end{table*}
	In this section, we provide details of the considered system model. The network is studied from a large-scale perspective based on stochastic geometry. The main mathematical notation related to the system model is summarized in Table \ref{Table1}.
	\subsection{Network model}
		We consider a multi-tier constellation setup for a downlink LEO satellite-based network as illustrated in Fig. \ref{systemmodelfigure}, where LEO satellites are deployed at $K$ different spherical surfaces concentric with the Earth, of altitudes $H_k$ above the mean sea level, with $H_k< H_{k+1}$, $\forall\ {} 1\le k \le K-1$.
		For facilitating the analytical tractability, we assume that at any certain instance, the locations and the orbits of LEO satellites are uncorrelated; while within each tier constellation, the LEO  satellites' locations are assumed to be distributed according to a homogeneous PPP $\Phi_k$ with intensity $\lambda_k=\frac{N_k}{4\pi(R_{\oplus}+H_k)^2}$, where $N_k$ is the average number of LEO satellites within the $k$-th tier; while $\Phi_i$ and $\Phi_j$ for $i\neq j$ are assumed to be independent \cite{Akram1,Akram2,Wang}. Therefore, the average number of LEO satellites within each tier constellation is given as $N_k=4\pi(R_{\oplus}+H_k)^2\lambda_k$, where $R_{\oplus}$ is the Earth radius; while $\Phi_i$ and $\Phi_j$ for $i\neq j$ are assumed to be independent \cite{Akram1,Akram2,Jung}. Furthermore, we assume that the locations of ground user equipments (gUEs) follow a uniform distribution. Without loss of generality and based on the Slivnyak’s theorem, we perform our analysis from the perspective of the typical gUE located at $(0,0,R_{\oplus})$ as illustrated in Fig. \ref{systemmodelfigure}, while the results hold for any gUE in the network \cite{Haenggi}. Let $x_{i,k}$ denote the location of the $i$-th LEO satellite that belongs to the $k$-th tier constellation, and $r_{i,k}$ represent the Euclidean distance from the $i$-th LEO satellite to the typical gUE. Additionally, let $x_{0,k}$ with $1\le k\le K$ depict the locations of the closest LEO satellites to the typical gUE from each tier constellation. Since  a gUE can only receive signals from the LEO satellites above its local horizon, the maximum distance between a gUE and a visible LEO satellite that belongs to the $k$-th tier constellation is given by $R_{{\rm max}_k}=\sqrt{H_k^2+2H_kR_{\oplus}}$, where $1\le k\le K$ as illustrated in Fig. \ref{systemmodelfigure} \cite{Okati,Okati2}.
		\vspace{-5mm}
	\begin{figure}[t]
		\centering
		\includegraphics[width=7cm,height=7cm]{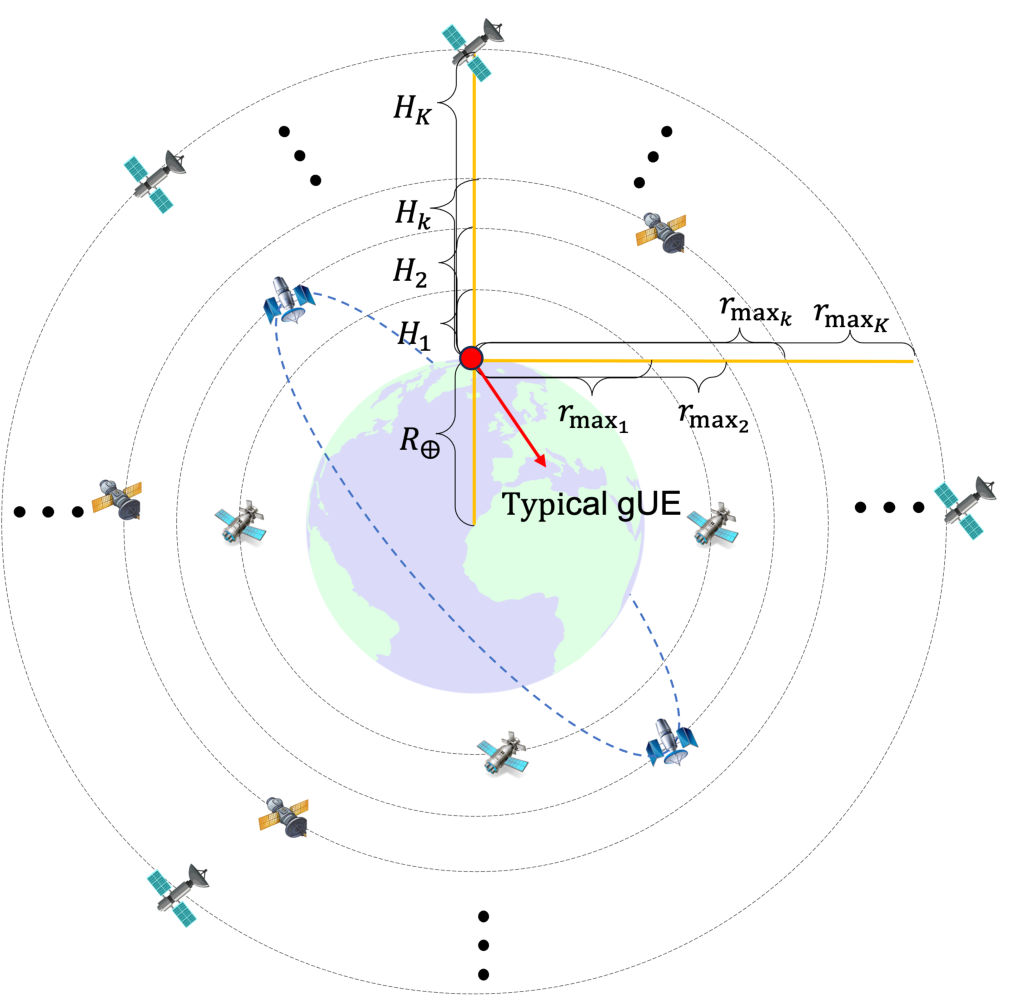}
		\caption{Network topology of a $K$-tier LEO satellites system.}
		\label{systemmodelfigure}
		\vspace{-8mm}
	\end{figure}
	\subsection{Satellite-to-ground channel model}
	We assume that the satellite-to-ground (StG) wireless channels experience both small-scale block fading and large-scale path-loss \cite{Akram1,Akram2,Jung,Jia,Okati2}. Regarding the small-scale fading, we adopt the well-known SR fading model for each link, which is typically a Rician fading channel with fluctuating line-of-sight (LoS) components  \cite{Jung}. Let $h_{i,j}$ denote the channel power gain of the link between the typical gUE and the $i$-th LEO satellite that belongs in the $j$-th tier.  Then,  the cumulative distribution function (CDF) of the channel power gain is given by
	\begin{align}\label{channelcdf}
		F_h(x)=&\left(\frac{2 b m}{2 b m+\Omega}\right)^m \sum_{z=0}^{\infty} \frac{(m)_z}{z ! \Gamma(z+1)}\left(\frac{\Omega}{2 b m+\Omega}\right)^z  \gamma\left(z+1, \frac{1}{2 b} x\right),
	\end{align}
	where $\Omega$ and $ 2b$ are the average power for the LoS and the multi-path components, respectively, and $ m$ is the fading parameters based on the Nakagami-$m$ fading channels \cite{Abdi}. Regarding the large-scale fading, we assume an unbounded singular path-loss model. More specifically, the path-loss effect of the StG channels only depends on the spatial distance between the LEO satellites and the gUEs; the path-loss between the typical gUE and a LEO satellite located at $x_{i,k}$ is given by $\ell(r_{i,k})=r_{i,k}^{-\alpha}$, where $\alpha>2$ is the path-loss exponent \cite{Talgat}. Finally, all StG wireless links exhibit additive white Gaussian noise with variance $\sigma^2$.

	We consider that both gUEs and LEO satellites are equipped with multiple antennas. For the sake of simplicity and tractability, we assume a widely-adopted simple sectorized model for the antenna directionality of LEO satellites and gUEs \cite{Marco,Park}. The adopted sectorized antenna model approximates the actual radiation pattern by a step function, where a constant main lobe gain is considered over the beamwidth as well as a constant side lobe gain is considered outside the beamwidth. This approach provides a suitable approximation of the actual beam pattern, while maintaining adequate accuracy for analysis purposes and does not deprive from the intuitions \cite{Marco,Park,Talgat}. Specifically, the antenna array gain of the LEO satellites that belongs in the $k$-th tier is given by  $G_k^L=\{G^L_{m,k},G^L_{s,k}\}$, where $G^L_{m,k}$ and $G^L_{s,k}$ represent the main- and the side-lobe gains of LEO satellites' antenna, respectively. Similarly, the antenna array gain of the gUE is characterized by two discrete values, i.e. the main-lobe gain $G^U_{m}$ and the side-lobe gain $G^U_{s}$. Furthermore, we assume that perfect beam alignment is achieved between each gUE and its serving LEO satellite by estimating the angle of arrival at both gUE and LEO satellite, such that the intended link between each gUE and its associated LEO satellite lies in the boresight direction of the antennas of both terminals. Additionally, we assume  the beam misalignment between each gUE and interfering LEO satellites. Therefore, the antenna array gain of the link between a gUE and the $k$-th tier LEO satellite is given by \cite{Park}
	\begin{equation}
		G=\begin{cases}
			G^L_{m,k}G^U_m, & \text{intended link},\\
			G^L_{s,k}G^U_s, & \text{interfering link}.
		\end{cases}
	\end{equation}
    It is worth mentioning that, although interfering satellites may also be in the boresight direction of the gUE's antenna, the employment of such assumption offers a tractable analysis for the considered large-scale LEO satellite networks. 
    \vspace{-5mm}
	\subsection{Power adjusting mechanism}\label{PAM}
	The altitudes of LEO satellites have a significant impact on the downlink network performance, e.g. higher altitudes of LEO satellites typically result in longer propagation distances of wireless signals, degrading the received signal quality at the gUEs. Moreover, according to International Telecommunication Union's regulations, the transmit power of LEO satellites is strictly constrained, such that the interference inflicted at existing GEO satellites systems is lower than a predefined threshold \cite{itu2003simulation,Hraishawi}. Motivated by the above discussion, we assume that a power adjusting mechanism is employed by the LEO satellites, where the transmit power of the LEO satellites of each tier is adjusted according to their altitudes, while the maximum achievable received signal power at a gUE from each tier LEO satellites are identical, i.e. $\frac{P_kG_{m,k}^L}{H_k^{\alpha}}=\frac{P_jG_{m,j}^L}{H_j^{\alpha}},\ {} \ {} \forall \ {} 1\le j,k\le K$, where $P_k$ is the transmit power of the $k$-th tier LEO satellites. It is worth mentioning that the adopted power adjusting mechanism not only captures the heterogeneity of multi-tier LEO satellite networks but also facilitates the analytical tractability of the developed framework.
	\section{Association schemes for Heterogeneous LEO satellite networks}\label{associationscheme}
	We propose three association schemes, which are associated with a two-stage procedure. In the first stage, $K$ LEO satellites are selected into the set of candidate LEO satellites; while in the second stage, the serving LEO satellite is selected among the set of candidate LEO satellites, according to the rules of the adopted association policy, i.e the Euclidean distance, the average received power over the fading channels, or the random selection\footnote{It should be noted that LEO satellites are operated on orbits with high velocity, and thus the potential handover process is required to keep associating the gUEs with specific serving LEO satellites. For simplicity, we assume that a soft-handover policy is employed to ensure the seamless connection between each gUE and its serving LEO satellite \cite{Akramhandover}.}. The three association schemes are presented and the analytical expression for the association probability for each scheme is derived.
	\vspace{-5mm}
	\subsection{Distance-based association (DbA) scheme}
	The DbA scheme is based on the Euclidean distance between the typical gUE and the LEO satellites. In specific, the DbA scheme enables each gUE to select and communicate with its nearest LEO satellite to maintain an acceptable received signal quality by minimizing the path-loss. Therefore, in the first stage of the DbA scheme, the nearest LEO satellites from each tier are selected into the set of candidate LEO satellites, i.e. $\mathcal{C}=\{x_{0,1},\cdots,x_{0,k},\cdots,x_{0,K}\}$. Then in the second stage, the typical gUE communicate with the closest LEO satellite from the set of candidate LEO satellites, e.g. $x_{0,k}=\argmin\nolimits_{x\in\mathcal{C}}\| x\|$ denotes the case where the typical gUE is associated with a LEO satellite belonging in the $k$-th tier. 
	It is worth mentioning that the DbA association scheme is a low-complexity scheme, which requires an a prior knowledge of the locations of LEO satellites, which can be obtained by monitoring the location of the LEO satellites through a low-rate feedback channel or by a global positioning system (GPS) \cite{YangGPS, HouraniGPS}. Hence, the DbA scheme holds significant utility for gUEs equipped with GPS capabilities.
	
	As the DbA relies on the Euclidean distance, we initially assess the distribution of the distance between the typical gUE and candidate LEO satellites, while the probability density function (PDF) of $r_{0,k}$ for $1\le k\le K$, is evaluated in the following lemma.
	\begin{Lemma}\label{contactdistance}
		The PDF of the distance between the typical gUE and its $k$-th candidate LEO satellite is given by
		\begin{align}\label{pdfcontactdistance}
			f_D(r,\lambda_k,H_k)=\frac{2\pi\lambda_kr(R_{\oplus}+H_k)}{R_{\oplus}}\exp\left(-\frac{\pi\lambda_k(R_{\oplus}+H_k)(r^2-H_k^2)}{R_{\oplus}}\right),
		\end{align}
		where $H_k\le r<R_{\mathrm{max}_k}$.
	\end{Lemma}
	\begin{proof}
		See Appendix \ref{proofcontactdistance}.
	\end{proof}
	It is worth emphasizing that while several previous studies have focused on deriving the Earth-centered zenith angle distribution to study single-tier LEO satellite-based networks\footnote{For a single-tier LEO satellite network, the distance between a gUE and a LEO satellite can be uniquely determined by the Earth-centered zenith angle \cite{Akram1,Akram2,Wang,Talgat}.}, the distance distribution presented in Lemma \ref{contactdistance} consists of a more general result and can be used for both single- and multi-tier LEO satellite networks. Specifically, by letting $H_k=H$ and $\lambda_k=\lambda$, $\forall\ {} 1\le k\le K$, the corresponding PDF for the special scenario, namely, the single-tier deployment is obtained.
	
	Since each gUE can only communicate with the LEO satellites that are located above the local horizon, the probability that at least one $k$-th tier LEO satellite is located above the local horizon of the typical gUE is computed as
	\begin{align}\label{pdfcontactdistancere}
		\hspace{-3mm}P^{\circ}(N_k,H_k)=1-\exp\left(-\frac{\pi\lambda_k(R_{\oplus}+H_k)(R_{\mathrm{max}_k}^2-H_k^2)}{R_{\oplus}}\right)=&1-\exp\left(-\frac{N_k}{2\left({R_{\oplus}}/{H_k}+1\right)}\right).
	\end{align}
	Note that $R_{\oplus}\approx 6371$ km and the altitude of the LEO satellites, i.e. $H_k$, is typically from $500$ km to $2000$ km. Moreover, for the considered large-scale LEO satellite networks, each tier consists of large number of LEO satellites, i.e. $N_k\gg 100$. Hence, the exponential term approaches to zero, i.e. $\exp\left(-\frac{N_k}{2\left({R_{\oplus}}/{H_k}+1\right)}\right)\rightarrow 0$, and $P^{\circ}(N_k,H_k)\rightarrow 1$. Therefore, we exclude the extreme scenario where no LEO satellite exists above the local horizon of the gUEs throughout our analysis \cite{Akram1}. The accuracy of the adopted assumption is validated by our numerical and simulation results presented in Section \ref{numsim}. Then, by applying the result from Lemma \ref{contactdistance}, we present the association probability for the DbA scheme, denoted as $\mathcal{A}_{D,k}$ in the following lemma.
	\begin{Lemma}\label{associationprobability}
		The probability of the typical gUE to communicate with a LEO satellite from the $k$-th tier, for the case where DbA scheme is adopted, is given by
		\begin{equation}\label{assoeq}
			\mathcal{A}_{D,k}=\sum\limits_{i=k}^{K}\int\nolimits_{H_{i}}^{H_{i+1}}\frac{2\pi\lambda_kr(R_{\oplus}+H_k)}{R_{\oplus}}\prod\limits_{j=1}^{i}\exp\left(-\frac{\pi\lambda_j(R_{\oplus}+H_j)(r^2-H_j^2)}{R_{\oplus}}\right)\mathrm{d}r,
		\end{equation}
		where $H_{K+1}=r_{{\rm max}_K}$.
	\end{Lemma}
	\begin{proof}
		See Appendix \ref{proofassociationprobability}.
	\end{proof}
	The results derived in Lemma \ref{associationprobability} provide valuable insights into the association probability of the DbA scheme and the interplay between various design parameters that affect the network performance. Moreover, it offers a more convenient and efficient way to evaluate the association probability than the Monte Carlo simulation, especially in the considered 3D large-scale multi-tier LEO satellite networks with a large number of network tiers. In specific, based on the analytical expression presented in Lemma \ref{associationprobability}, the association probability can be immediately determined for any given network parameters, including the number of tiers, the height of each LEO satellite tier, and the density of LEO satellites, while the exponential product term in \eqref{assoeq} shows the interdependence among LEO satellite tiers. Additionally, by substituting $\lambda_k$ with $\frac{N_k}{4\pi(R_{\oplus}+H_k)}$, we have, 
	$\mathcal{A}_{D,k}=\sum\nolimits_{i=k}^{K}\int\nolimits_{H_{i}}^{H_{i+1}}\frac{N_k r}{2R_{\oplus}(R_{\oplus}+H_k)}\prod\nolimits_{j=1}^{i}\exp\left(-\frac{\pi\lambda_j(R_{\oplus}+H_j)(r^2-H_j^2)}{R_{\oplus}}\right)\mathrm{d}r$, based on which we can easily observe that the association probability of each tier increases with the average number of LEO satellites and decreases with the altitude of this tier.
	\subsection{Power-based association (PbA) scheme}
	The second association scheme, namely PbA scheme, takes into consideration of the average received signal power at the gUEs. In specific, the PbA scheme associates each gUE with the LEO satellite that provides the highest average signal power. Note that LEO satellites which belong in the same tier have identical transmit power and antenna array gain, and thus the nearest one to the gUE provides the strongest average received signal. Hence, the set of candidate LEO satellites is identical with that formulated for the DbA scheme. In the second stage, the serving LEO satellite is selected among the set of candidate LEO satellites, which provides the highest average received signal power instead of shortest Euclidean distance.
	
	We now focus on the association probability for the PbA scheme. Let $\xi_{0,k}$ represent the average received signal power from the LEO satellite located at $x_{0,k}$. Note that $\xi_{0,k}$ is computed by taking the expectation over the channel power gain, i.e. $\xi_{0,k}=\bar{h} P_kG_{m,k}^L r_{0,k}^{-\alpha}G_m^U$, where $\bar{h}$ is the first-order moment of the channel power gain, i.e. $\bar{h}=(2b+\Omega)$ for the SR fading channel model \cite{Abdi}. In the following lemma, we provide the PDF of $\xi_{0,k}$, which is essential to assess the association probability for the PbA scheme.
	\begin{Lemma}\label{AveragePowerPdf}
		The PDF of the average received signal power at the typical gUE from its $k$-th candidate LEO satellite is given by
		\begin{align}\label{eqAveragePowerPdf}
			f_P(\xi,\lambda_k,H_k)=\frac{2\pi\lambda_k(R_{\oplus}+H_k)}{\xi\alpha R_{\oplus}}\left(\frac{\varrho_k}{\xi}\right)^{\frac{2}{\alpha}}\exp\left(-\frac{\pi\lambda_k(R_{\oplus}+H_k)}{R_{\oplus}}\left(\left(\frac{\varrho_k}{\xi}\right)^{\frac{2}{\alpha}}-H_k^2\right)\right),
		\end{align}
		where $\varrho_k=\bar{h}P_kG^L_{m,k}G_m^U$,  $\xi_{\mathrm{min}_k}<\xi\le\xi_{\mathrm{max}_k}$, ${\xi}_{\mathrm{min}_k}={\varrho_k}{R_{\mathrm{max}_k}^{-\alpha}}$ and ${\xi}_{\mathrm{max}_k}={\varrho_k}{H_k^{-\alpha}}$.
	\end{Lemma}
	\begin{proof}
		See Appendix \ref{ProofAveragePowerPdf}.
	\end{proof}
	It is worth mentioning that, while the distance distribution was studied to analyze the performance of the DbA scheme, the power distribution presented in Lemma \ref{AveragePowerPdf} offers a more appropriate method for assessing the performance achieved by the PbA scheme. In particular, the average received signal power from the candidate LEO satellites exhibits identical lower and upper bounds. The minimum value of the average received signal power from any candidate LEO satellite approaches zero due to the extremely long communication distance, i.e., $\xi_{\mathrm{min}_k}\approx 0$, $\forall ~ 1\le k\le K$; whereas the maximum average received signal power from all candidate LEO satellites is also identical, i.e., $\xi_{\mathrm{max}_k}=\xi_{\mathrm{max}_j}$, $\forall ~ 1\le j,k\le K$, owing to the power adjusting mechanism introduced in Section \ref{PAM}. In following lemma, we provide the association probability for the PbA scheme.
	\begin{Lemma}\label{AssociationProbabilityPbA}
		The probability of the typical gUE to communicate with a LEO satellite from the $k$-th tier, for the case where PbA scheme is adopted, is given by
		\begin{equation}\label{assoeq2}
			\hspace{-1.1mm}\mathcal{A}_{P,k}=\frac{2\pi\lambda_k\varrho_k^{\frac{2}{\alpha}}(R_{\oplus}+H_k)}{\alpha R_{\oplus}}\int\nolimits_{{\xi}_{\mathrm{min}_k}}^{{\xi}_{\mathrm{max}_k}}\hspace{-1mm}\left(\frac{1}{\xi}\right)^{\frac{2+\alpha}{\alpha}}\prod_{j=1}^{K}\exp\left(-\frac{\pi\lambda_j(R_{\oplus}+H_j)}{R_{\oplus}}\left(\hspace{-2mm}\left(\frac{\varrho_{k}}{\xi}\right)^{\frac{2}{\alpha}}\hspace{-2mm}-H_j^2\right)\hspace{-2mm}\right)\mathrm{d}\xi.
		\end{equation}
	\end{Lemma}
	\begin{proof}
		See Appendix \ref{ProofAssociationProbabilityPbA}.
	\end{proof}
	From the expression derived in above lemma, we can observe that the PbA scheme considers one additional parameter compared to the DbA scheme, i.e. the received signal power; whereas the analytical expression for the association probability of the PbA scheme has a more concise form than DbA, due to the employment of the power adjusting mechanism.
	\vspace{-5mm}
	\subsection{Random selection-based association (RbA) scheme}
	The RbA scheme is based on a random tier selection. In specific, the first stage process of the RbA scheme is the same with the DbA and PbA schemes. In the second stage, the typical gUE randomly selects one serving LEO satellite among the set of candidate LEO satellites. The RbA scheme does not require any operations for comparing distance and/or received power in the second stage process,  leading to low implementation/signaling complexity. On one hand, the RbA scheme can be adopted in scenarios with a large number of network tiers, where its simplicity offers significant advantages. On the other hand, for scenarios with a small number of network tiers, the RbA scheme exhibits almost the same complexity as the DbA and the PbA schemes. Furthermore, as the RbA enables a gUE to randomly select one LEO satellite among $K$ candidate LEO satellites, the probability of a gUE being associated with each tier is equal, i.e. $
	\mathcal{A}_{R,k}=\mathcal{A}_{R,j}=\frac{1}{K},~ \forall\ {} 1\le k,j\le K$.
	\section{Downlink Performance of LEO satellite networks}\label{downper}
	We investigate the downlink performance of the considered heterogeneous LEO satellite networks in the context of proposed association schemes. Specifically, the downlink performance is initially evaluated in terms of the coverage probability, corresponding to the probability that the instantaneous SINR exceeds a predefined threshold.   Additionally, we discuss the performance of the proposed association schemes under another two distinct extreme scenarios, namely, Rayleigh fading and non-fading environments. Finally,  we investigate the spectral efficiency achieved at typical gUEs with the employment of each association scheme.
	\subsection{Coverage performance}
	The analytical expressions for the coverage probability achieved by each association scheme with SR, Rayleigh, and non-fading channels are presented in the following subsections.
	\subsubsection{Coverage probability under the DbA scheme}
	In order to facilitate the analysis of the coverage probability achieved by the DbA scheme, we evaluate the distribution of the contact distance, i.e. the distance between the typical gUE and its serving LEO satellite. We present the conditional  PDF of the contact distance for the DbA scheme in the following lemma, by considering the case where the typical gUE is associated with the $k$-th tier network.
	\begin{Lemma}\label{conditionaldistancepdf}
		When the DbA scheme is employed, the PDF of the contact distance is given by
		\begin{align}
			\hspace{-3mm}	f_{D,k}(r|\mathcal{A}_{D,k})=&\frac{2\pi\lambda_k(R_{\oplus}+H_k)}{R_{\oplus}\mathcal{A}_{D,k}}\sum_{i=k}^{K}\mathbbm{I}(H_i,H_{i+1})\exp\left(-\sum_{j=1}^{i}\frac{\pi\lambda_j(R_{\oplus}+H_j)(r^2-H_j^2)}{R_{\oplus}}\right)r,
		\end{align}
		where $\mathbb{I}(H_i,H_{i+1})$ is the indicator function, given by 
		\begin{equation}
			\mathbb{I}(H_i,H_{i+1})=\begin{cases}
				1,& \text{ $H_i\le r<H_{i+1}$},\\
				0,&\text{\rm otherwise}.
			\end{cases}
		\end{equation}
	\end{Lemma}
	\begin{proof}
		See Appendix \ref{proofconditionaldistancepdf}.
	\end{proof}
	It is worth noting that the results presented in Lemma \ref{conditionaldistancepdf} are applicable to various LEO satellite network topologies. For instance, by setting $H_k=H$, $G_{m,k}^L=G_m^L$, $P_k=P$ and $\lambda_k=\lambda$ $\forall~1\le k\le K$, we capture the single-tier LEO satellite networks as investigated in previous studies \cite{Akram1, Akram2, Jung}.
	
	By aiming to derive the closed-form expression for the coverage probability, we approximate the power of the aggregate interference signals with its mean, and the coverage probability achieved by the DbA scheme is evaluated in the following theorem.
	\begin{Theorem}\label{coverageprobability}
		The coverage probability achieved by the DbA scheme is given by
		\begin{equation}
			\mathcal{P}_D(\beta)=\sum\limits_{k=1}^{K}\mathcal{A}_{D,k}\mathcal{P}_{D,k}(\beta),
		\end{equation}
		where
		\begin{align}
			&\mathcal{P}_{D,k}(\beta)\approx\int\nolimits_{H_k}^{R_{\mathrm{max}_k}}\left(1-F_h\left(\frac{\beta r^{\alpha}\big({\widetilde{\mathcal{I}}_D}(r)+\sigma^2\big)}{ P_{k}G^L_{m,k}G^U_m}\right)\right)f_{D,k}(r|\mathcal{A}_{D,k})\mathrm{d}r,
		\end{align}
	and
		\begin{align}
			{\widetilde{\mathcal{I}}_D}(r)=\sum\limits_{j=1}^{K}\frac{2\pi\lambda_j\bar{h}P_jG_{s,j}^LG_s^U(R_{\oplus}+H_j)(\max\{\min\{r^{2-\alpha},H_j^{2-\alpha}\}-R_{\mathrm{max}_j}^{2-\alpha},0\})}{(\alpha-2)R_{\oplus}}.
		\end{align}
	\end{Theorem}
	\begin{proof}
		See Appendix \ref{proofcoverageprobability}.
	\end{proof}
	
	Note that the term $\widetilde{\mathcal{I}}_D(r)$ represents the conditional mean of the aggregate interference, which is a function of random variable $r$. Such methodology takes into account the randomness of spatial deployments and fading channels, and avoids the intractability that typically arises when performing system-level analysis with the SR fading channel model. The accuracy of the adopted approximation is validated by our numerical and simulation results in Section \ref{numsim}. Given that $F_h(\cdot)$ is a monotonically increasing function, we can infer, based on Theorem \ref{coverageprobability}, that the shorter contact distance, i.e. $r$, the higher transmit power, i.e. $P_k$, the smaller interference i.e. $\widetilde{\mathcal{I}}_D(r)$, lead to a higher converge probability. Moreover,  the closed-form expression of $\widetilde{\mathcal{I}}_D(r)$ provides a convenient approach for analyzing the average interference levels in the large-scale LEO satellite networks. Specifically, we can easily observe that the average interference received by the typical gUE is directly proportional to the density and transmit power of LEO satellites, while a greater path-loss exponent mitigates the interference power.  
	\subsubsection{Coverage probability under the PbA scheme}
	We evaluate the distribution of the average received signal power at the typical gUE from its serving LEO satellite. By considering the scenario where the typical gUE is associated with the $k$-th tier, the conditional PDF of $\xi_{0,k}$ is evaluated in the following lemma.
	\begin{Lemma}\label{ConditionAveragePowerPdf}
		When the PbA scheme is employed, the PDF of the average received signal power at the typical gUE from its serving LEO satellite is given by
		\begin{align}\label{eqConditionAveragePowerPdf}
			\hspace{-3mm}f_{P,k}(\xi|\mathcal{A}_{P,k})=\frac{2\pi\lambda_k\varrho_k^{\frac{2}{\alpha}}(R_{\oplus}+H_k)}{\alpha R_{\oplus}\mathcal{A}_{P,k}}\left(\frac{1}{\xi}\right)^{\frac{2}{\alpha}+1}\prod_{j=1}^{K}\exp\left(-\frac{\pi\lambda_j(R_{\oplus}+H_j)}{R_{\oplus}}\left(\hspace{-2mm}\left(\frac{\varrho_k}{\xi}\right)^{\frac{2}{\alpha}}-H_j^2\right)\hspace{-2mm}\right).
		\end{align}
	\end{Lemma}
	\begin{proof}
		See Appendix \ref{ProofConditionAveragePowerPdf}.
	\end{proof}   
	By utilizing the aforementioned results, we provide the analytical expression for the coverage probability achieved by the PbA scheme, in  the following theorem.
	\begin{Theorem}\label{coverageprobability2}
		The coverage probability achieved by the PbA scheme is given by
		\begin{equation}
			\mathcal{P}_P(\beta)=\sum\limits_{k=1}^{K}\mathcal{A}_{P,k}\mathcal{P}_{P,k}(\beta),
		\end{equation}
		where
		\begin{align}
			&\mathcal{P}_{P,k}(\beta)\approx\int\nolimits_{\xi_{\mathrm{min}_k}}^{\xi_{\mathrm{max}_k}}\left(1-F_h\left( \frac{\beta\bar{h}\big(\widetilde{\mathcal{I}}_{P}(\xi_{0,k})+\sigma^2\big)}{\xi_{0,k}}\right)\right)f_{P,k}(\xi_{0,k}|\mathcal{A}_{P,k})\mathrm{d}\xi_{0,k},
		\end{align}
		and
		\begin{align}
			\widetilde{\mathcal{I}}_P(\xi_{0,k})=\sum\limits_{j=1}^{K}\frac{2\pi\lambda_jP_jG_{s,j}^LG_s^U\bar{h}(R_{\oplus}+H_j)\big(\min\big\{(\frac{\varrho_k}{\xi_{0,k}})^{\frac{2-\alpha}{\alpha}},H_j^{2-\alpha}\big\}-R_{\mathrm{max}_j}^{2-\alpha}\big)}{(\alpha-2)R_{\oplus}}.
		\end{align}
	\end{Theorem}
	\begin{proof}
		See Appendix \ref{proofcoverageprobability2}.
	\end{proof}
	We can easily observe from Theorem \ref{coverageprobability2} that the interference level experienced by the typical gUE using the PbA scheme is lower compared to that with the DbA scheme. This observation can be justified by the fact that the PbA scheme prioritizes the selection of LEO satellites based on the received signal power. Therefore, the gUE is more likely to select a satellite with a stronger signal, which in turn reduces the amount of interference from other LEO satellites.
	\subsubsection{Coverage probability under the RbA scheme}
	By considering that the typical gUE is associated with a LEO satellite belonging in the $k$-th tier, the PDF of the contact distance with the employment of the RbA scheme is presented in the following lemma.
	\begin{Lemma}
		When the RbA scheme is employed, the PDF of contact distance is given by
		\begin{align}
			f_{R,k}(r|\mathcal{A}_{R,k})=\frac{2\pi\lambda_kr(R_{\oplus}+H_k)}{R_{\oplus}}\exp\left(-\frac{\pi\lambda_k(R_{\oplus}+H_k)(r^2-H_k^2)}{R_{\oplus}}\right).
		\end{align}
	\end{Lemma}
	\begin{proof}
		The CDF of the contact distance by conditioning on the case that the typical gUE is associated with the $k$-th tier LEO satellite is given by
		\begin{align}
			F_{R,k}(r|\mathcal{A}_{R,k})
			=1-\mathbb{P}\left\{r_{0,k}\ge r|\text{ the typical gUE is associated with the $k$-th tier}\right\}.
		\end{align}
		Since the tier selection is based on a random operation, i.e. it is independent with $r_{0,k}$, we have $F_{R,k}(r|\mathcal{A}_{R,k})=1-\mathbb{P}\left\{r_{0,k}\ge r\right\}$. By directly following the results presented in Lemma \ref{contactdistance}, the final expression of $f_{R,k}(r|\mathcal{A}_{R,k})$ is derived.
	\end{proof}
	We now present the analytical expression for the coverage probability achieved by the RbA scheme in the following theorem.
	\begin{Theorem}\label{coverageprobability3}
		The coverage probability achieved by the RbA scheme is given by
		\begin{equation}
			\mathcal{P}_R(\beta)=\sum\limits_{k=1}^{K}\mathcal{A}_{R,k}\mathcal{P}_{R,k}(\beta),
		\end{equation}
		where
		\begin{align}
			&\mathcal{P}_{R,k}(\beta)\approx\int\nolimits_{H_k}^{R_{\mathrm{max}_k}}\left(1-F_h\left(\frac{\beta r_{0,k}^{\alpha}(\widetilde{\mathcal{I}}_R(r_{0,k})+\sigma^2)}{ P_{k}G^L_{m,k}G^U_m}\right)\right)f_{R,k}(r_{0,k}|\mathcal{A}_{R,k})\mathrm{d}r_{0,k},\\
			\hspace{-10mm}\text{and}&\nonumber\\
			&\widetilde{\mathcal{I}}_R(r_{0,k}) = \frac{1}{(\alpha-2)R_{\oplus}}\bigg({2\pi \lambda_kP_kG_{s,k}^LG_s^U\bar{h}(R_{\oplus}+H_k)(r_{0,k}^{2-\alpha}-R_{\mathrm{max}_k}^{2-\alpha})}\nonumber\\&\hspace{4.5cm}+\sum\limits_{j=1,j\neq k}^{K}{2\pi\lambda_jP_jG_{s,j}^LG_s^U\bar{h}(R_{\oplus}+H_j)(H_j^{2-\alpha}-R_{\mathrm{max}_j}^{2-\alpha})}\bigg).
		\end{align}
	\end{Theorem}
	\begin{proof}
		The proof follows similar steps presented in Appendix \ref{proofcoverageprobability}, but with the average interference power and association probability that correspond to the RbA scheme.
	\end{proof}
	Some key insights can be obtained from the expression presented in the above theorem. Although the RbA scheme has a low-complexity implementation, it results in a higher level of interference experienced by the gUEs compared to the other two schemes. Specifically, the random tier selection results in some nearby LEO satellites from other tiers contributing significant interference to the typical gUE.
	\subsubsection{Special fading scenarios}
	By aiming to evaluate the effects of fading channels on the heterogeneous LEO satellites networks, we investigate the achieved performance in another two fading scenarios, namely, Rayleigh and non-fading channels. Rayleigh fading model is used to capture the worst scenario with weak or obstructed LoS links, due to the near-ground blockages \cite{Akram1}. Although the Rayleigh fading channel has been discussed in \cite{Okati} for the single-tier LEO satellite-based networks with simplistic distance-based association scheme, we now evaluate the coverage probability with Rayleigh fading channels in terms of each proposed association scheme. The exact analytical expression for the coverage probability achieved by the each association scheme in Rayleigh fading scenario is provided in following Proposition.
	\begin{Proposition}\label{RayleighDbA}
		The coverage probabilities achieved by the DbA, PbA and RbA schemes with Rayleigh fading channels are given by
		\begin{equation}
			\mathcal{P}^{\dagger}_{\Delta}(\beta)=\sum_{k=1}^{K}\mathcal{A}_{\Delta,k}\mathcal{P}^{\dagger}_{\Delta,k}(\beta),
		\end{equation}
		where $\Delta\in\{D,P,R\}$ stands for the DbA, PbA and RbA schemes, respectively, $\mathcal{P}^{\dagger}_{\Delta,k}(\beta)$ is given by
		\begin{equation}
			\mathcal{P}^{\dagger}_{\Delta,k}(\beta)=\begin{cases}
				\int_{H_k}^{R_{\mathrm{max}_k}}\mathcal{L}^{\dagger}_{\mathcal{I}_D}(\varphi_D)\exp(-\varphi_D\sigma^2)f_{D,k}(r_{0,k}|\mathcal{A}_{D,k})\mathrm{d}r_{0,k},& \textit{\bf if $\Delta=D$ },\\
				\int_{\xi_{\mathrm{min}_k}}^{\xi_{\mathrm{max}_k}}\mathcal{L}^{\dagger}_{\mathcal{I}_P}(\varphi_P)\exp(-\varphi_P\sigma^2)f_{P,k}(\xi_{0,k}|\mathcal{A}_{P,k})\mathrm{d}\xi_{0,k},& \textit{\bf if $\Delta=P$ },\\
				\int_{H_k}^{R_{\mathrm{max}_k}}\mathcal{L}^{\dagger}_{\mathcal{I}_R}(\varphi_R)\exp(-\varphi_R\sigma^2)f_{R,k}(r_{0,k}|\mathcal{A}_{R,k})\mathrm{d}r_{0,k},& \textit{\bf if $\Delta=R$ },
			\end{cases}
		\end{equation}
		$\mathcal{L}^{\dagger}_{\mathcal{I}_{\Delta}}(\varphi_{\Delta})$ is given by
		\begin{equation}\label{ray1}
			\mathcal{L}^{\dagger}_{\mathcal{I}_{\Delta}}(\varphi_{\Delta})=
			\begin{cases}
				\prod\limits_{j=1}^{K}\exp\left(-\frac{2\pi\lambda_jP_jG_{s,j}^LG_s^U\varphi_D(R_{\oplus}+H_j)}{{R_{\oplus}}(\alpha-2)/(\varpi_j(\max\{r_{0,k},H_j\},\varphi_D)-\varpi_j(R_{\mathrm{max}_j},\varphi_D))}\right),& \textit{\bf if $\Delta=D$},\\
				\prod\limits_{j=1}^{K}\exp\Big(-\frac{2\pi\lambda_jP_jG_{s,j}^LG_s^U\varphi_P(R_{\oplus}+H_j)}{{R_{\oplus}}(\alpha-2)/(\varpi_j(\max\{({\varrho_k}/{\xi_{0,k}})^{1/\alpha},H_j\},\varphi_P)-\varpi_j(R_{\mathrm{max}_j},\varphi_P))}\Big),&  \textit{\bf if $\Delta=P$ },\\
				\exp\left(-\frac{2\pi\lambda_kP_kG_{s,k}^LG_s^U\varphi_R(R_{\oplus}+H_k)}{{R_{\oplus}}(\alpha-2)/(\varpi_k(r_{0,k},\varphi_R)-\varpi_k(R_{\mathrm{max}_k},\varphi_R))}\right)\times\\
				\hspace{2cm}\prod\limits_{j=1,j\neq k}^{K}\exp\left(-\frac{2\pi\lambda_jP_jG_{s,j}^LG_s^U\varphi_R(R_{\oplus}+H_j)}{{R_{\oplus}}(\alpha-2)/(\varpi_j(H_{j},\varphi_R)-\varpi_j(R_{\mathrm{max}_j},\varphi_R))}\right),&  \textit{\bf if $\Delta=R$},
			\end{cases}
		\end{equation} $\varpi_j(r,\varphi)={r}^{2-\alpha}\ {}_2F_1\left[1,\frac{\alpha-2}{\alpha};2-\frac{2}{\alpha};-P_jG_{s,j}^LG_s^U r^{-\alpha}\varphi\right]$; $\varphi_D=\varphi_R=\frac{\beta}{P_kG_{m,k}^LG_m^Ur_{0,k}^{-\alpha}}$, $\varphi_P=\frac{\beta}{\xi_{0,k}}$.
	\end{Proposition}
	\begin{proof}
		See Appendix \ref{RayleighDbAproof}.
	\end{proof}
	In contrast to the previous discussion with the Rayleigh fading channel, in the following proposition, we investigate another extreme case, i.e. non-fading scenario. In this case, the channel remains constant, i.e. $h=1$; this scenario provides a useful benchmark to evaluate the impact of the fading effect on the performance of LEO satellites networks.  
	\begin{Proposition}\label{nonfadingDbA}
		The coverage probabilities achieved by the DbA, PbA and RbA schemes with non-fading channels are given by
		\begin{align}
			\mathcal{P}^{\ddagger}_{\Delta}(\beta)=\sum_{k=1}^{K}\mathcal{A}_{\Delta,k}\mathcal{P}^{\ddagger}_{\Delta,k}(\beta),
		\end{align}
		where 
		\begin{equation}\label{nfc1}
			\mathcal{P}^{\ddagger}_{\Delta,k}(\beta)=\begin{cases}
				\int\nolimits_{H_k}^{R_{\mathrm{max}_k}}\Big(\frac{1}{2}-\frac{1}{\pi}\int\nolimits_{0}^{\infty}\frac{\Im\left\{\psi(t,\mathcal{S}_D)\exp(it\beta\sigma^2)\right\}}{ t}\mathrm{d}t\Big)f_{D,k}(r_{0,k}|\mathcal{A}_{D,k})\mathrm{d}r_{0,k}, &\textit{\bf if $\Delta=D$},\\
				\int\nolimits_{\xi_{\mathrm{min}_k}}^{\xi_{\mathrm{max}_k}}\Big(\frac{1}{2}-\frac{1}{\pi}\int\nolimits_{0}^{\infty}\frac{\Im\left\{\psi(t,\mathcal{S}_P)\exp(it\beta\sigma^2)\right\}}{t}\mathrm{d}t\Big)f_{P,k}(\xi_{0,k}|\mathcal{A}_{P,k})\mathrm{d}\xi_{0,k}, & \textit{\bf if $\Delta=P$},\\
				\int\nolimits_{H_k}^{R_{\mathrm{max}_k}}\Big(\frac{1}{2}-\frac{1}{\pi}\int\nolimits_{0}^{\infty}\frac{\Im\left\{\psi(t,\mathcal{S}_R)\exp(it\beta\sigma^2)\right\}}{ t}\mathrm{d}t\Big)f_{R,k}(r_{0,k}|\mathcal{A}_{R,k})\mathrm{d}r_{0,k}, &\textit{\bf if $\Delta=R$},
			\end{cases}
		\end{equation}
		$i=\sqrt{-1}$, $\psi(t,\mathcal{S}_{\Delta})=\exp(-it\mathcal{S}_{\Delta})\mathcal{L}^{\ddagger}_{\mathcal{I}_{\Delta}}(-it\beta)$, $\mathcal{S}_{D}=\mathcal{S}_{R}=P_kG_{m,k}^LG_m^Ur_{0,k}^{-\alpha}$, $\mathcal{S}_{P}=\xi_{0,k}$, $\mathcal{L}^{\ddagger}_{\mathcal{I}_{\Delta}}(z)$ is given by 
		\begin{equation}
			\mathcal{L}^{\ddagger}_{\mathcal{I}_{\Delta}}(z)=
			\begin{cases}
				\prod\limits_{j=1}^{K}\exp\Big(\frac{\pi\lambda_j(R_{\oplus}+H_j)}{{R_{\oplus}}\alpha/(\vartheta_j(\max\{r_{0,k},H_j\},z)-\vartheta_j(R_{\mathrm{max}_j},z))}\Big), & \textit{\bf if $\Delta=D$},\\
				\prod\limits_{j=1}^{K}\exp\Big(\frac{\pi\lambda_j(R_{\oplus}+H_j)}{{R_{\oplus}}\alpha/(\vartheta_j\big(\max\{({\varrho_k}/{\xi_{0,k}})^{1/\alpha},H_j\},z\big)-\vartheta_j(R_{\mathrm{max}_j},z))}\Big),& \textit{\bf if $\Delta=P$},\\
				\exp\left(\frac{\pi\lambda_k(R_{\oplus}+H_k)}{{R_{\oplus}}\alpha/(\vartheta_k(r_{0,k},z)-\vartheta_k(R_{\mathrm{max}_k},z))}\right)\prod\limits_{{j=1,j\neq k}}^{K}\hspace{-0mm}\exp\Big(\frac{\pi\lambda_j(R_{\oplus}+H_j)}{{R_{\oplus}}\alpha/(\vartheta_j(H_j,z)-\vartheta_j(R_{\mathrm{max}_j},z))}\Big),& \textit{\bf if $\Delta=R$},
			\end{cases}
		\end{equation}	$\vartheta_j(r,z)=P_jG_{s,j}^LG_s^U\left(r^2\left(\alpha-2\mathcal{E}\left(\frac{2+\alpha}{\alpha},r^{-\alpha}z\right)\right)+2z^{\frac{2}{\alpha}}\Gamma\left[-\frac{2}{\alpha}\right]\right)$, and $\mathcal{E}(a,b)=\int_{1}^{\infty}\frac{e^{-bx}}{x^a}\mathrm{d}x$.
	\end{Proposition}
	\begin{proof}
		See Appendix \ref{nonfadingDbAproof}.
	\end{proof}
	\subsection{Spectral efficiency}
	We assess the spectral efficiency achieved by each gUE. Specifically, the spectral efficiency is measured in terms of the average ergodic rate per unit bandwidth \cite{Jo,Kalamkar}; for simplicity, the average ergodic rate is computed in unit of nats/s/Hz. By following a similar approach with the derivation of the coverage probability, the average ergodic rate achieved at the typical gUE is given by 
	\begin{equation}\label{se}
		{\mathcal{R}}_{\Delta}=\sum_{k=1}^{K}\mathcal{R}_{\Delta,k}\mathcal{A}_{\Delta,k},
	\end{equation}
	where $\mathcal{R}_{\Delta,k}$ denotes the average ergodic rate per unit bandwidth of a typical gUE associated with the $k$-tier. By considering that the typical gUE is associated with the $k$-th tier, an analytical expression for $\mathcal{R}_{\Delta,k}$ is presented in the following proposition.
	\begin{Proposition}
		The average ergodic rate per unit bandwidth of the typical gUE associated with the $k$-th tier is given by
		\begin{align}
			\mathcal{R}_{\Delta,k}=\int_{0}^{\infty}\frac{\mathcal{P}_{\Delta,k}(\beta)}{1+\beta}\mathrm{d}\beta,\label{sek}
		\end{align}
		where $\Delta\in\{D,P,R\}$ and $\mathcal{P}_{\Delta,k}(\beta)$ is given in Theorem \ref{coverageprobability}, \ref{coverageprobability2}\& \ref{coverageprobability3}.
		\begin{proof}
			The proof directly follows from the definition of the average ergodic rate and the spectral efficiency  \cite{Jo,Kalamkar}. 
		\end{proof}
	\end{Proposition}
	Hence, by combining \eqref{se} and \eqref{sek}, the spectral efficiency of the network under each association scheme is obtained. Additionally, we can obtain the spectral efficiency of the network under the Rayleigh fading and non-fading scenarios, by replacing $\mathcal{P}_{\Delta,k}(\beta)$ with $\mathcal{P}^{\dagger}_{\Delta,k}(\beta)$ and $\mathcal{P}^{\ddagger}_{\Delta,k}(\beta)$, respectively.

	\begin{table}[t]
		\centering
		\caption{SIMULATION PARAMETERS.}
		\renewcommand{\arraystretch}{1}
		\scalebox{0.88}{
			\begin{tabular}{ | m{4.5cm} | m{1.8cm}| m{7cm} | m{3.5cm}|} 
				\hline
				{\bf Parameters}& {\bf Values}  &{\bf Parameters}& {\bf Values}\\ 
				\hline
				Network tiers ($K$) & $3$ &Altitude of LEO satellites ($H_1,H_2,H_3$)&$500$ km, $600$ km, $700$ km\\ 
				\hline
				Path-loss exponent ($\alpha$) & $3$ &Average number of LEO satellites ($N_1,N_2,N_3$)&$500$, $1000$, $1500$\\ 
				\hline
				Radius of  Earth ($R_{\oplus}$)&$6371$ km &Transmit power of LEO satellites ($P_1,P_2,P_3$)&$32$ W, $55.3$ W, $87.8$ W\\ 
				\hline
				Thermal noise power ($\sigma^2$)&$-81.4$ dBm &Main-lobe gain of LEO satellites $(G_{m,1}^L,G_{m,2}^L,G_{m,3}^L)$&$47$ dBi, $47$ dBi, $47$ dBi \\ 
				\hline
				Main-lobe gain of gUEs $\left(G_m^U\right)$ & $10$ dBi& Side-lobe gain of LEO satellites $(G_{s,1}^L,G_{s,2}^L,G_{s,3}^L)$&$27$ dBi, $27$ dBi, $27$ dBi\\ 
				\hline
				Side-lobe gain of gUEs $\left(G_s^U\right)$ & $0$ dBi  &Parameters of SR fading model $(b,\Omega,m)$&$(0.158,19.4,1.59)$\\ 
				\hline
		\end{tabular}}
		\vspace{-5mm}
		\label{simulation}
	\end{table}
	\section{Numerical and Simulation Results}\label{numsim}
	We present analytical and simulation results to validate the accuracy of our model and illustrate the performance of the proposed association schemes. Unless otherwise stated, we use the parameters given in Table \ref{simulation}. It is worth mentioning that, the proposed analytical framework is generic, and the selection for these parameter values is for the purpose of illustration \cite{Talgat}.
	
	\begin{figure}[t]
		\centering
		\subfigure[Coverage probability for DbA scheme versus the SINR threshold.]{
			\includegraphics[width=0.3\linewidth]{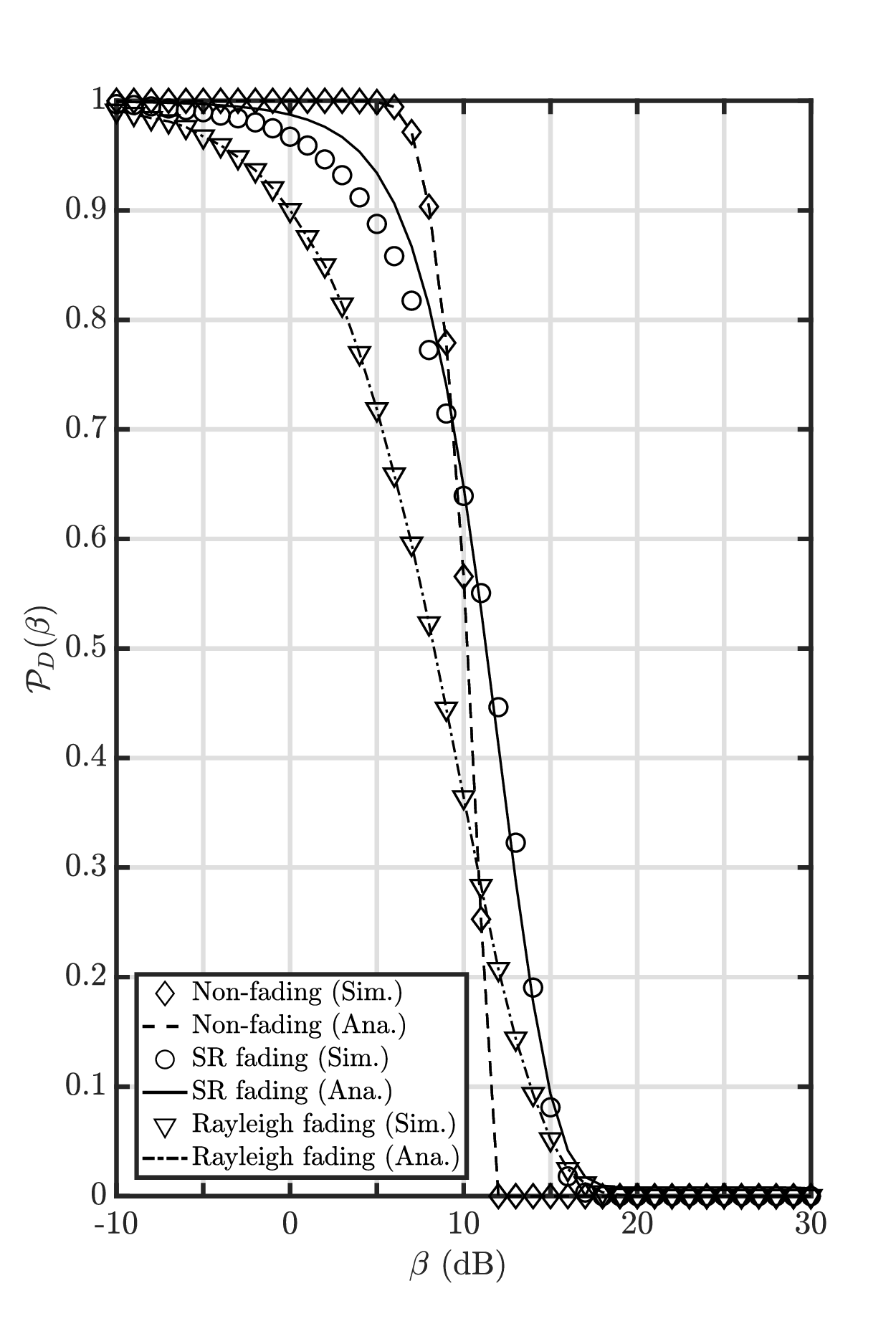}	
			\label{subfig1}
		}
		\hspace{1mm}
		\subfigure[Coverage probability for PbA scheme versus the SINR threshold.]{
			\includegraphics[width=0.3\linewidth]{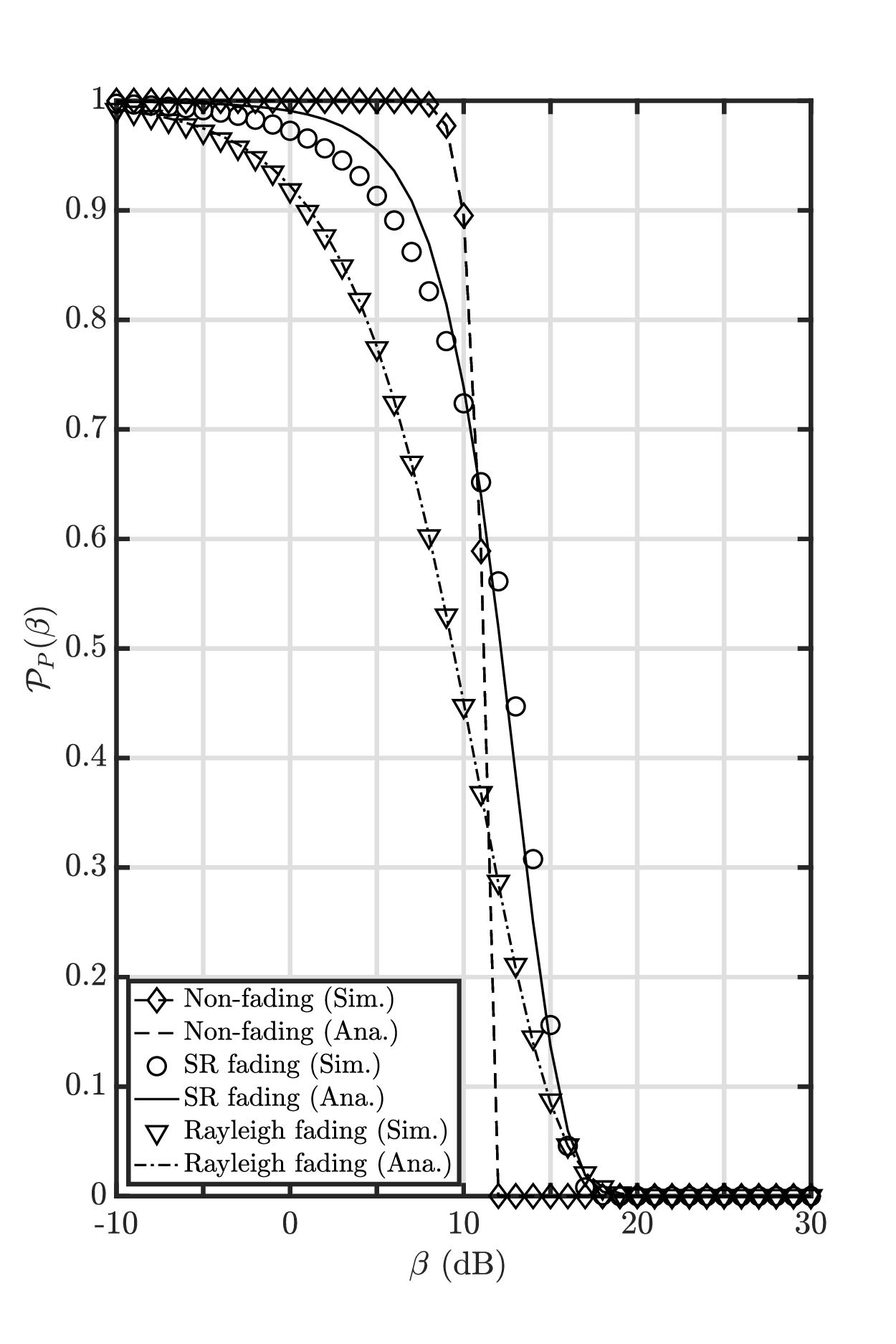}
			\label{subfig2}
		}
		\hspace{1mm}
		\subfigure[Coverage probability for RbA scheme versus the SINR threshold.]{
			\includegraphics[width=0.3\linewidth]{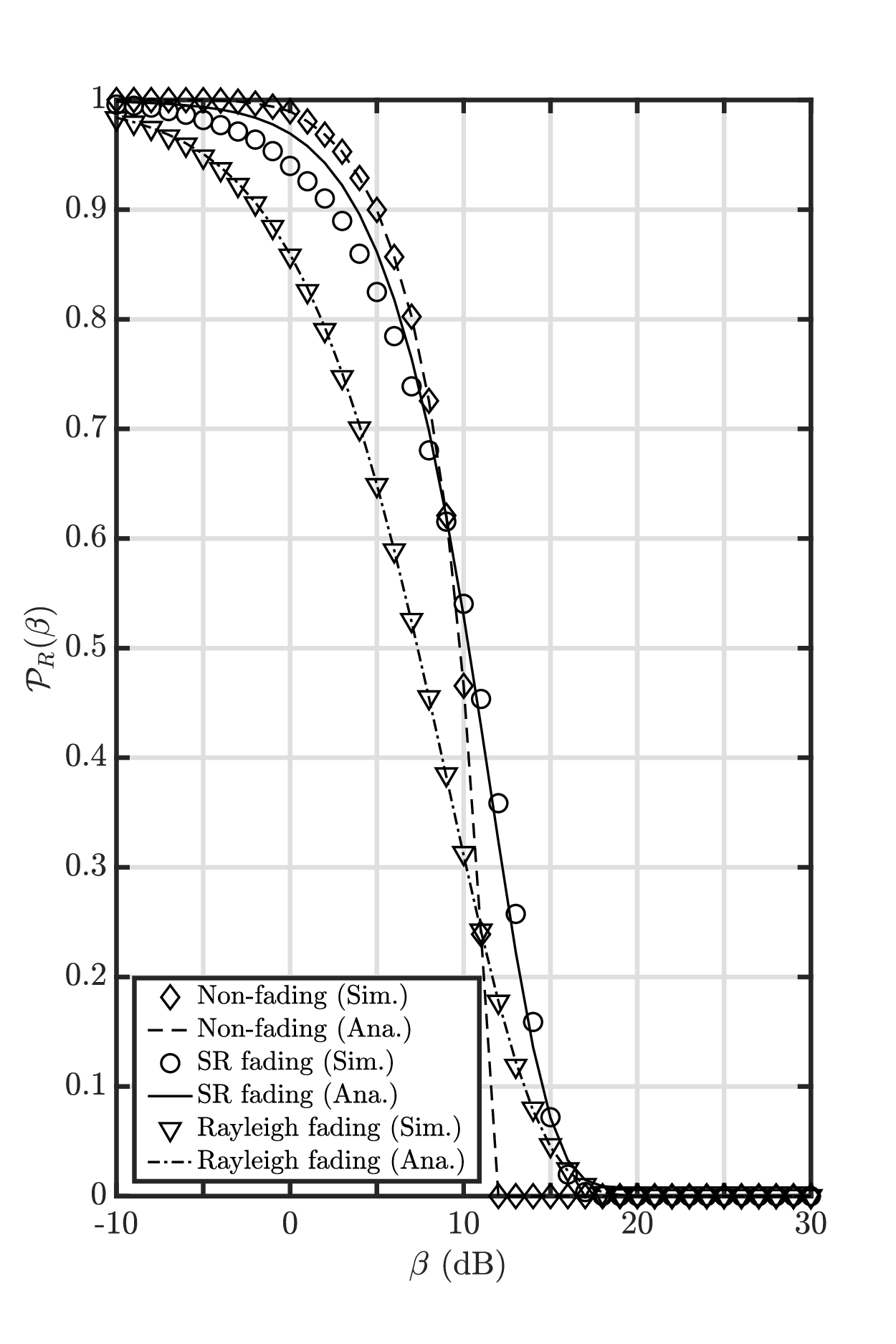}
			\label{subfig3}
		}
		\caption{Coverage probability achieved by each association scheme, in the Shadow Rician, the Rayleigh and the non-fading scenarios.}
		\label{coveragefig}
		\vspace{-8mm}
	\end{figure}
	
	Fig. \ref{coveragefig} illustrates the impact of fading channels on the coverage performance of LEO satellite networks. In specific, Fig. \ref{coveragefig} plots the coverage probability achieved by the three proposed association schemes, i.e. DbA, PbA, and RbA, in different fading scenarios that are denoted as ``Non-fading'', ``SR fading'', and ``Rayleigh fading'' . For the low SINR  regime, we observe that  the absence of fading results in a higher coverage probability in comparison to the cases with  fading, for all association schemes. However, as the SINR threshold increases, the decrease in the performance of the non-fading scenario is greater than the SR and  Rayleigh fading, leading to the observation that non-fading and SR fading achieves the worst and best performance at high SINRs, respectively. This can be justified by the fact that fading channel leads to the random attenuation of signals, including interference, which can consequently enhance the SINR in certain operation scenarios. Additionally, the non-fading scenario achieves a better performance at low SINR thresholds due to its stable signal strength. Moreover, the SR fading channel with its strong LoS component, corresponds to a higher signal quality than the Rayleigh fading channel, resulting in a higher SINR. Finally, we can observe that for each association scheme, the coverage performance achieved under the Rayleigh and non-fading scenarios demonstrate a perfect match between simulation (denoted as ``Sim.") and analysis (denoted as ``Ana.") results;
	the adopted approximations for the analysis of the SR fading scenario only lead to minor gaps between the simulation and analytical results at low SINR thresholds.
	\begin{figure}[t]
		\centering
		\includegraphics[width=0.8\linewidth]{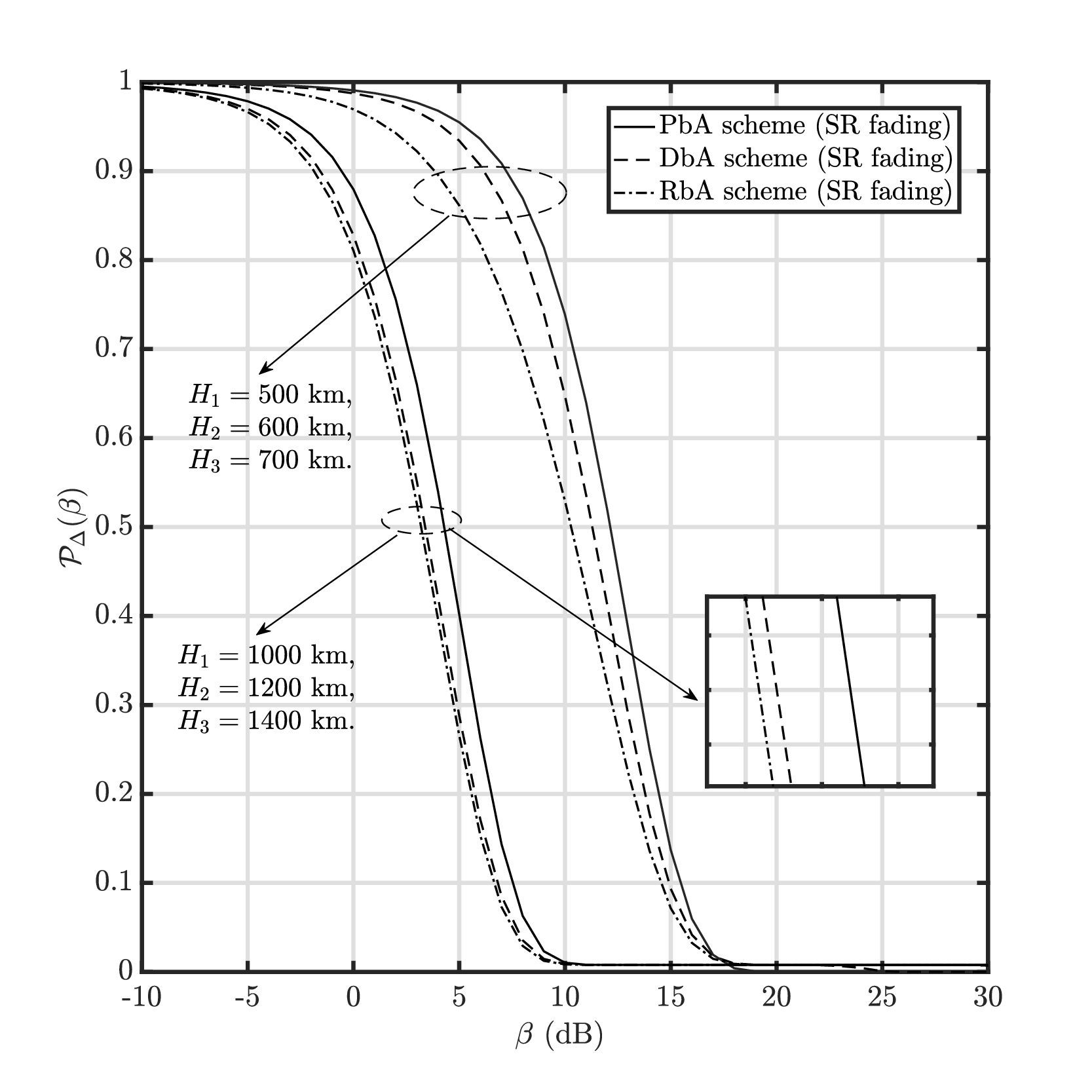}
		\caption{Coverage probability achieved by each association scheme versus the SINR threshold for different LEO satellites' altitudes.}
		\label{altitudeimpact}
	\end{figure}
	
	Fig. \ref{altitudeimpact} depicts the effect of LEO satellites' altitudes on the coverage performance. Specifically, Fig. \ref{altitudeimpact} plots the coverage probability versus the SINR threshold for different association schemes and different altitudes. Initially, it can be observed that the coverage probability decreases with the increase of the LEO satellites' altitudes. This was expected since higher altitudes of LEO satellites correspond to longer propagation distances between each gUE and its serving satellites, which results in a lower SINR. Moreover, we observe that the PbA scheme surpasses both the DbA and RbA schemes in terms of the coverage probability. This can be justified by the fact that the PbA scheme associates the gUE with the LEO satellite that delivers the highest average signal power, thereby boosting the SINR. In contrast, the DbA scheme associates the gUE with the nearest LEO satellite, that may lead to severe interference due to the existence of nearby satellites with larger transmit power, resulting in a degraded SINR at the gUE. Moreover, the RbA scheme randomly associates the gUE with one of the candidate LEO satellites, leading to the worst coverage probability. Finally, the difference in coverage probability between the DbA and RbA schemes diminishes as the satellites' altitude increases. This is due to the fact that as the altitude of the satellites increases, the difference in path-loss between the nearest and other candidate LEO satellites becomes less significant, and thus the advantage of associating the gUE with the nearest LEO satellite in the DbA scheme becomes negligible. 
	
	\begin{figure}[t]
		\centering
		\includegraphics[width=0.8\linewidth]{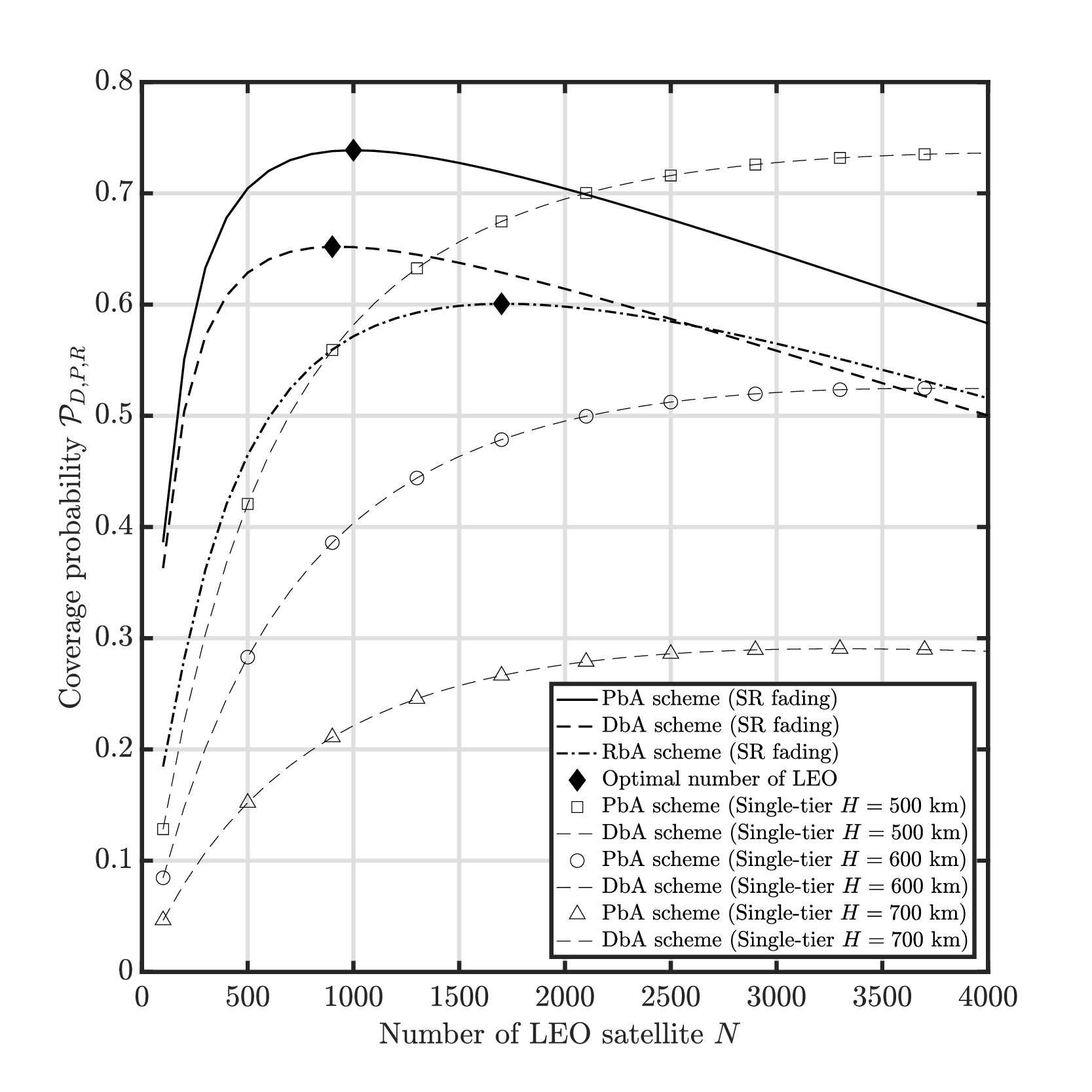}
		\caption{The coverage probability achieved by each association schemes versus the average number of LEO satellites, where $\beta=10$ dB and $N_1=N_2=N_3=N$.  }
		\label{numimp}
	\end{figure}
	Fig. \ref{numimp} demonstrates the effect of the average number of LEO satellites on the coverage performance. For the sake of simplicity, we consider that each tier comprises an equal average number of LEO satellites, i.e. $N_1=N_2=N_3=N$. Accordingly, Fig. \ref{numimp} plots the coverage probability versus the number of LEO satellites within each tier, for different association schemes, with the SINR threshold set to $\beta=10$ dB. We can clearly observe that for all association schemes, as the number of LEO satellites increases, there is an initial increase in coverage probability, which subsequently declines upon exceeding a specific threshold (denoted as ``Optimal number of LEO"). The observed enhancement in coverage probability can be explained by the decreasing distance between the serving LEO satellite and the gUE, which results in decreased path-loss and subsequently leads to an elevated SINR at the gUE. However, when the number of LEO satellites exceeds a certain limit, the increased number of interfering satellites leads to a severe multi-user interference, and thus to a reduced coverage probability. Another notable observation is that the PbA scheme consistently achieves the highest coverage probability among the three association schemes. Additionally, the DbA scheme outperforms the RbA scheme when the number of LEO satellites is relatively small, while as the number of LEO satellites increases, the RbA scheme surpasses the DbA scheme in terms of coverage probability. This observation can be explained by the fact that, the RbA scheme allows equal probability of association across tiers, increasing the chance of associating a gUE to satellites with higher transmit power from higher tiers, as opposed to DbA scheme which prioritizes proximity, resulting in frequent associations with lower-tier satellites with lower transmit power. Finally, we can observe that in the conventional single-tier network, the PbA and the DbA result in the same coverage probability. While the multi-tier network typically attains a higher coverage probability than the single-tier network, it is worth noting that with a large number of LEO satellites, the single-tier network can actually surpass it in coverage probability. This is based on the fact that the multi-tier network offers more access options than the single-tier network, which also leads to additional interference and thus finally results in a lower coverage probability than the single-tier network in scenario with dense deployments of LEO satellites.
	
	\begin{figure}[t]
		\centering
		\includegraphics[width=0.8\linewidth]{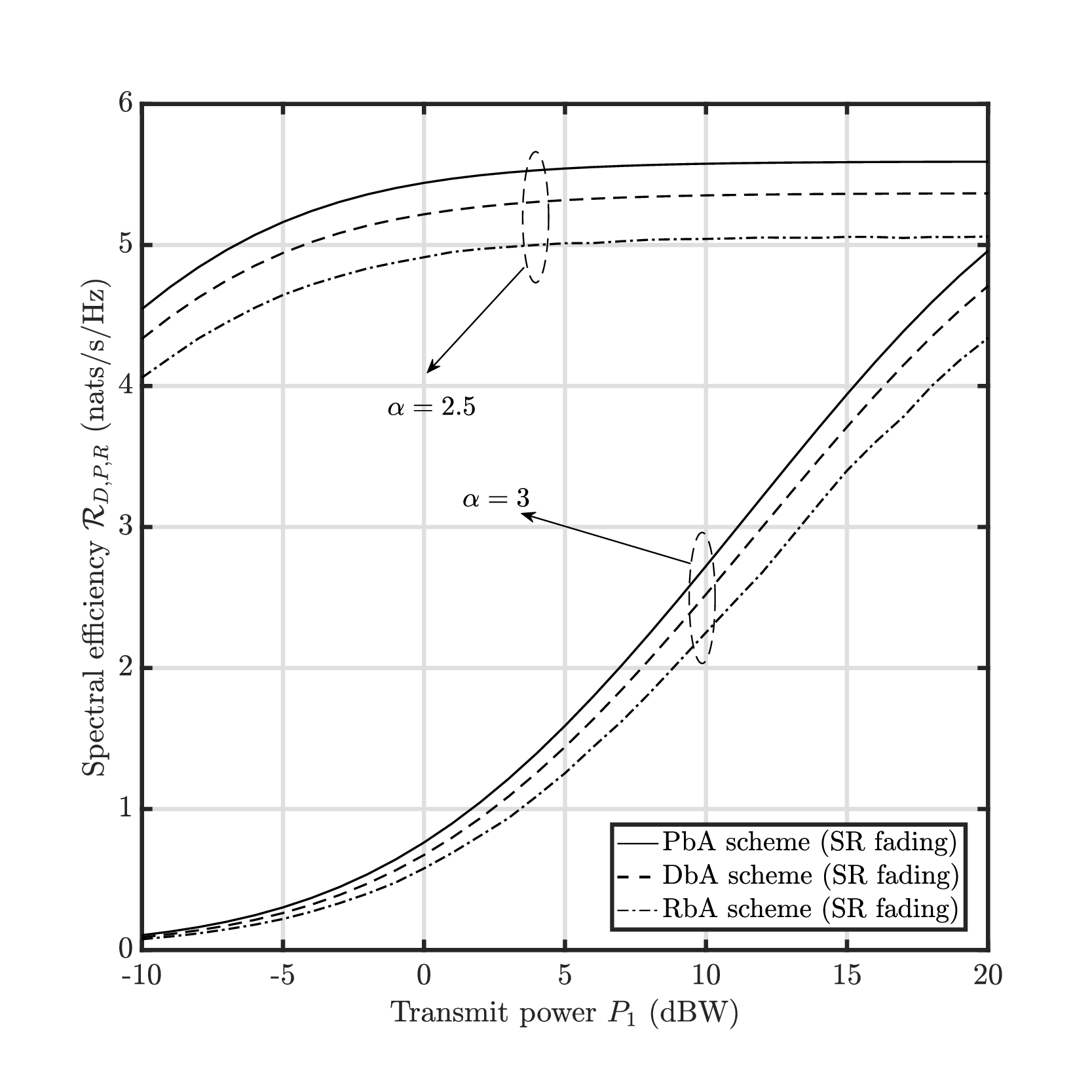}
		\caption{The spectral efficiency versus the transmit power of LEO satellites, where $P_2=\frac{H_2^{\alpha}G_{m,1}^L}{H_1^{\alpha}G_{m,2}^L}P_1$ and $P_3=\frac{H_3^{\alpha}G_{m,1}^L}{H_1^{\alpha}G_{m,3}^L}P_1$. }
		\label{datarate}
		\vspace{-8mm}
	\end{figure}
	
	Finally, Fig. \ref{datarate} shows the impact of LEO satellites' transmit power and the large-scale fading on the spectral efficiency of the considered networks. More specifically, Fig. \ref{datarate} plots the spectral efficiency versus the transmit power of the first tier LEO satellites for different path-loss exponent.  Firstly, for all association schemes, i.e. the DbA, the PbA, and the RbA schemes, a higher path-loss corresponds to a lower spectral efficiency, which can be attributed to the reduced signal strength at the receiver. Secondly, the PbA scheme consistently outperforms the DbA and RbA schemes, which is in line with the remarks presented in Fig. \ref{altitudeimpact} and Fig. \ref{numimp}. Finally, when the path-loss exponent is small (e.g., $\alpha=2.5$), increasing the transmit power initially leads to a higher data rate. However, the gUE experiences strong interference from neighbouring satellites as a result of the reduced signal attenuation, which counterbalances the benefits of higher transmit power, causing the data rate to plateau and remain constant at a saturation point.
	\section{Conclusion}\label{conclu}
	In this work, we present a comprehensive analysis of large-scale heterogeneous LEO satellite networks based on stochastic geometry. We introduce three association schemes and evaluate their performance in terms of association probability, coverage probability, and spectral efficiency, while analytical expressions are derived by leveraging tools from stochastic geometry. By aiming to reveal the impact of fading channels on network performance, we evaluate the coverage probability in different fading scenarios, i.e. SR, Rayleigh and non-fading channel models. Our results highlight the critical impact of fading channels, altitude, the number of satellites, and transmit power on the coverage probability and spectral efficiency. Specifically, we demonstrate that  the PbA scheme consistently outperforms the other two schemes, while the RbA scheme surpasses the DbA scheme as the number of satellites increases; the optimal number of LEO satellites for achieving the highest coverage probability is numerically demonstrated. We show that higher path-loss results in reduced spectral efficiency, and increasing transmit power initially enhances spectral efficiency which then remains constant due to the interference from neighbouring satellites. Our work offers valuable insights and guidance for the design and optimization of future large-scale heterogeneous LEO satellite networks.
	\appendices
	\section{Proof of Lemma \ref{contactdistance}}\label{proofcontactdistance}
	The proof follows a similar approach as \cite{Akram1}. We first compute the CDF of  $r_{0,k}$ based on the null probability of PPP \cite{Haenggi}, i.e.
	\begin{equation}\label{nullprob}
		\begin{aligned}
			&F_D(r,\lambda_k,H_k)=\mathbb{P}\{r_{0,k}\le r\}=1-\mathbb{P}\{r_{0,k}> r\}
			=1-\exp\left(-\lambda_k \Psi(r,H_k)\right),
		\end{aligned}
	\end{equation}
	\noindent where  $	\Psi(r,H_k)=\frac{\pi(r^2-H_k^2)(R_{\oplus}+H_k)}{R_{\oplus}}$ represents the area of the spherical dome as shown in Fig. \ref{proofappfig}.
	Finally, by taking the derivative of $F_D(r,\lambda_k,H_k)$ with respect to $r$, i.e. $f_D(r,\lambda_k,H_k)=\frac{\partial}{\partial r}F_D(r,\lambda_k,H_k)$, the final expression in Lemma \ref{contactdistance} can be derived.
	\begin{figure}[t]
		\centering
		\includegraphics[width=0.8\linewidth]{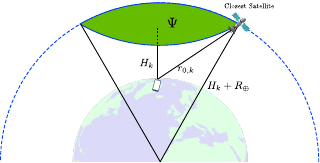}
		\caption{The closest LEO satellite of the $k$-th tier constellation. }
		\label{proofappfig}
		\vspace{-8mm}
	\end{figure}
	\section{Proof of Lemma \ref{associationprobability}}\label{proofassociationprobability}
	By employing the DbA scheme, the probability that the typical gUE is associated with the $k$-th tier LEO satellite can be expressed as $\mathcal{A}_{D,k} = \mathbb{P}\left\{r_{0,k}<\min_{j, j\neq k}r_{0,j}\right\}$. Note that for the considered heterogeneous LEO satellites networks, the orbit altitude is much smaller than the Earth radius, i.e. $H_k\ll R_{\oplus}$, and thus it is easy to verify that $H_{1}\le H_{2}\le \cdots \le H_{K} \le R_{\mathrm{max}_1}\le  R_{\mathrm{max}_2}\le\cdots\le R_{\mathrm{max}_K}$. Hence, by applying the law of total probability and by denoting $R_{{\rm max}_K}$ with $H_{K+1}$, we can have 
	\begin{equation}
		\mathcal{A}_{D,k}=\sum_{i=k}^{K}\mathbb{P}\left\{\min\limits_{{1\le j\le K,  j\neq k}} r_{0,j}> r_{0,k}~\&~ H_{i}\le r_{0,k}\le H_{i+1}\right\}.
	\end{equation}
	Upon observing that $r_{0,k}\le \min\limits_{1\le j\le K-i} r_{0,i+j}$ is valid for $H_{i}\le r_{0,k}\le H_{i+1}$, we can further simplify $\mathcal{A}_{D,k}$ as 
	\begin{equation}
		\mathcal{A}_{D,k}
		=\sum_{i=k}^{K}\mathbb{P}\left\{\min_{{1\le j\le i, j\neq k}}r_{0,j}> r_{0,k}~\&~ H_{i}\le r_{0,k}\le H_{i+1} \right\}.
	\end{equation} 
	Then, by noticing that $r_{0,j}$ and $r_{0,k}$, $\forall\ {}j\neq k$ are mutually independent, each term of the probability within $\mathcal{A}_{D,k}$ can be evaluated as
	\begin{equation}
		\begin{aligned}\label{AppB3}
			&\mathbb{P}\left\{\min_{{1\le j\le i, j\neq k}}r_{0,j}> r_{0,k}~\&~ H_{i}\le r_{0,k}\le H_{i+1} \right\}\\
			=&\int\nolimits_{H_{i}}^{H_{i+1}}{\prod_{j=1,j\neq k}^{i}\exp\left(-\frac{\pi\lambda_j(R_{\oplus}+H_j)(r^2-H_j^2)}{R_{\oplus}}\right)}f_D(r,\lambda_k,H_k)\mathrm{d}r\\
			=&\int\nolimits_{H_{i}}^{H_{i+1}}\prod_{j=1}^{i}\exp\left(-\frac{\pi\lambda_j(R_{\oplus}+H_j)(r^2-H_j^2)}{R_{\oplus}}\right)\frac{2\pi\lambda_kr(R_{\oplus}+H_k)}{R_{\oplus}} \mathrm{d}r,
		\end{aligned}
	\end{equation}
	where the last step is obtained by substituting the expression of $f_D(r,\lambda_k,H_k)$ given in \eqref{pdfcontactdistance} and by cancelling out the common factors of the numerator and denominator. 
	\section{Proof of Lemma \ref{AveragePowerPdf}}\label{ProofAveragePowerPdf}
	To begin with, we evaluate the CDF of $\xi_{0,k}$ as following
	\begin{equation}
		\begin{aligned}\label{CDFaveragepower}
			F_{P}(\xi,\lambda_k,H_k)
			=&\mathbb{P}\{\xi_{0,k}\le \xi\}=1-\mathbb{P}\left\{r_{0,k}<(\bar{h}P_kG^L_{m,k}G^U_m/\xi)^{\frac{1}{\alpha}}\right\}\\=&1-F_D\big((\bar{h}P_kG^L_{m,k}G^U_m/\xi)^{\frac{1}{\alpha}},\lambda_k,H_k\big).
		\end{aligned}
	\end{equation}
	The PDF of $\xi_{0,k}$ is obtained by taking the derivative of $F_P(\xi,\lambda_k,H_k)$ with respect to $\xi$, i.e. $f_P(\xi,\lambda_k,H_k)=\frac{\partial}{\partial \xi}F_P(\xi,\lambda_k,H_k)$.
	\section{Proof of Lemma \ref{AssociationProbabilityPbA}}\label{ProofAssociationProbabilityPbA}
	According to the procedure of the PbA scheme, the probability that the typical gUE is associated with the $k$-th tier LEO satellite is formulated as $\mathcal{A}_{P,k}=\mathbb{P}\left\{\xi_{0,k}>\max\limits_{{1\le j\le K, j\neq k}}\xi_{0,j}\right\}$.	The deployments of $K$ tiers of LEO satellites follow $K$ independent PPPs, which indicates that $\xi_{0,j}$ for $1\le j\le K$ are mutually independent between each other. Hence, based on the order statistics \cite{yang2011order}, we have
	\begin{align}
		\mathcal{A}_{P,k}=&\prod_{j=1,j\neq k}^{K}\mathbb{P}\left\{\xi_{0,k}>\xi_{0,j}\right\}=\mathbb{E}_{\xi_{0,k}}\left\{\prod_{j=1,j\neq k}^{K}\Big(1-F_P\left(\xi_{0,k},\lambda_j,H_j\right)\Big)\right\}\nonumber\\=&\mathbb{E}_{\xi_{0,k}}\left\{\prod_{j=1,j\neq k}^{K}\exp \left(-\pi\lambda_j\frac{R_{\oplus}+H_j}{R_{\oplus}}\left(\left(\frac{\bar{h}P_jG^L_{m,j}G_m^U}{\xi_{0,k}}\right)^{\frac{2}{\alpha}}-H_j^2\right)\right)\right\}.
	\end{align}
	Finally, by evaluating the expectation over $\xi_{0,k}$, the result in Lemma \ref{AssociationProbabilityPbA} is derived.
	\section{Proof of Lemma \ref{conditionaldistancepdf}}\label{proofconditionaldistancepdf}
	Without loss of generality, we consider that the typical gUE is associated with the $k$-th tier network for $1\le k\le K$.  Then, the CDF of the contact distance is computed as
	\begin{align}\label{eqq}
		F_{D,k}(r|\mathcal{A}_{D,k})=&\mathbb{P}\left\{r_{0,k}\le r\ {}\Big|\ {}r_{0,k}<\min_{\substack{1\le j\le K,\\  j\neq k}}r_{0,j}\right\}
		=1-\frac{\mathbb{P}\left\{r_{0,k}> r~\& ~ r_{0,k}<\min\limits_{{1\le j\le K,  j\neq k}}r_{0,j}\right\}}{\mathbb{P}\left\{r_{0,k}<\min\limits_{{1\le j\le K,  j\neq k}}r_{0,j}\right\}}.
	\end{align}
	Note that $\mathbb{P}\left\{r_{0,k}<\min_{1\le j\le K, j\neq k}r_{0,j}\right\}=\mathcal{A}_{D,k}$. Then, for the case where $H_i\le r<H_{i+1}$, the probability term in \eqref{eqq} is computed as
	\begin{align}\label{eqq2}
		&\mathbb{P}\left\{r< r_{0,k}<\min\limits_{{1\le j\le K,  j\neq k}}r_{0,j} ~\&~H_i\le r<H_{i+1}\right\}\nonumber\\
		&=\int\nolimits_{r}^{H_{i+1}}\prod_{j=1}^{i}\exp\left(-\frac{\pi\lambda_j(R_{\oplus}+H_j)(r^2-H_j^2)}{R_{\oplus}}\right)\frac{2\pi\lambda_kr(R_{\oplus}+H_k)}{R_{\oplus}}  \mathrm{d}r.
	\end{align}
	By taking into account all possible ranges of $r$ within $H_k\le r\le R_{\mathrm{max}_k}$ and by substituting \eqref{eqq2} into \eqref{eqq}, we obtain the final expression of $F_{D}(r|\mathcal{A}_{D,k})$, i.e.
	\begin{align}
		&F_{D,k}(r|\mathcal{A}_{D,k})\nonumber\\=&1-\frac{\sum_{i=k}^{K}\mathbbm{I}(H_i\le r<H_{i+1})\int\nolimits_{r}^{H_{i+1}}\prod_{j=1}^{i}\exp\left(-\frac{\pi\lambda_j(R_{\oplus}+H_j)(r^2-H_j^2)}{R_{\oplus}}\right)\frac{2\pi\lambda_kr(R_{\oplus}+H_k)}{R_{\oplus}}  \mathrm{d}r}{\mathcal{A}_{D,k}}.
	\end{align}
	Finally, the PDF of the contact distance, i.e. $f_{D,k}(r|\mathcal{A}_k)$, can be derived by taking the derivative of $F_{D,k}(r|\mathcal{A}_{D,k})$ with respect to $r$, i.e. $f_{D,k}(r|\mathcal{A}_k)=\frac{\partial}{\partial r}F_{D,k}(r|\mathcal{A}_{D,k})$.
	\section{Proof of Theorem \ref{coverageprobability}}\label{proofcoverageprobability}
	By the law of the total probability and according to the definition of the coverage probability, we have $\mathcal{P}_{D}(\beta)=\sum_{k=1}^{K}\mathcal{P}_{D,k}(\beta)\mathcal{A}_{D,k}$, where $\mathcal{P}_{D,k}(\beta)$ represents the coverage probability when the typical gUE is associated with the $k$-th tier, which is given by
	\begin{equation}
	\begin{aligned}\label{AppD1}
		\mathcal{P}_{D,k}(\beta)=&\mathbb{P}\left\{\frac{P_kh_{0,k}G_{m,k}^LG^U_mr_{0,k}^{-\alpha}}{\sum\nolimits_{x_{i,j}\in\cup_{j=1}^{K}\Phi_j^{\backslash x_{0,k}}} P_jh_{i,j}r_{i,j}^{-\alpha}G^L_{s,j}G^U_s+\sigma^2}\ge\beta\right\}\\
		=&\mathbb{E}_{r_{0,k},h_{i,j}}\left\{\mathbb{P}\left\{h_{0,k}\ge\frac{\beta\big(\sum\nolimits_{x_{i,j}\in\cup_{j=1}^{K}\Phi_j^{\backslash x_{0,k}}} P_jh_{i,j}r_{i,j}^{-\alpha}G^L_{s,j}G^U_s+\sigma^2\big)}{P_kG^L_{m,k}G^U_mr_{0,k}^{-\alpha}}\right\}\right\}\\
		\approx&\mathbb{E}_{r_{0,k}}\left\{\mathbb{P}\left\{h_{0,k}\ge\frac{\beta\big(\widetilde{\mathcal{I}}_{D}(r_{0,k})+\sigma^2\big)}{P_kG^L_{m,k}G^U_mr_{0,k}^{-\alpha}}\right\}\right\}
		=\mathbb{E}_{r_{0,k}}\left\{1-F_h\left(\frac{\beta\big(\widetilde{\mathcal{I}}_{D}(r_{0,k})+\sigma^2\big)}{P_kG^L_{m,k}G^U_mr_{0,k}^{-\alpha}}\right)\right\},
	\end{aligned}
	\end{equation}
	where $F_h(\cdot)$ is the CDF of channel power gain; $\widetilde{\mathcal{I}}_D(r_{0,k})$ is the mean of the aggregate interference power, and is computed as following
	\begin{equation}
	\begin{aligned}\label{eqd1}
		\widetilde{\mathcal{I}}_D(r_{0,k})=&\mathbb{E}\left\{\sum\nolimits_{x_{i,j}\in\cup_{j=1}^{K}\Phi_j^{\backslash x_{0,k}}} P_jG^L_{s,j}h_{i,j}r_{i,j}^{-\alpha}G^U_s\right\}\\
		\overset{(a)}{=}&\sum_{j=1}^{K}\int\nolimits_{0}^{2\pi}\int\nolimits_{\theta} \bar{h}\lambda_jP_j G^L_{s,j}r_{i,j}^{-\alpha}G^U_s(R_{\oplus}+H_j)^2\sin \theta\mathrm{d}\theta\mathrm{d}\phi\\
		=&\sum_{j=1}^{K}\int\nolimits_{\theta}2\pi\bar{h}\lambda_j  P_jG^L_{s,j}r_{i,j}^{-\alpha}G^U_s(R_{\oplus}+H_j)^2\sin \theta\mathrm{d}\theta,
	\end{aligned}
	\end{equation}
	where $(a)$ follows from the Campbell's Theorem of PPP with the spherical coordinates \cite{Haenggi,Akram1}. Then, by replacing $\theta$ with a function of $r$ based on the law of cosines, i.e. 
	\begin{equation}
		\theta=\arccos\left(\frac{R_{\oplus}^2+(R_{\oplus}+H_j)^2-r^2}{2R_{\oplus}(R_{\oplus}+H_j)}\right),
	\end{equation}
	we have
	\begin{equation}\label{eqd2}
		\begin{aligned}
			\mathrm{d}\theta=&\frac{r}{R_{\oplus} (H_j+R_{\oplus}) \sqrt{1-\frac{\left((H_j+R_{\oplus})^2-r^2+R_{\oplus}^2\right)^2}{4 R_{\oplus}^2 (H_j+R_{\oplus})^2}}}\mathrm{d}r,\\
			\sin\theta=&\sqrt{1-\frac{\left((H_j+R_{\oplus})^2-r^2+R_{\oplus}^2\right)^2}{4 R_{\oplus}^2 (H_j+R_{\oplus})^2}}.
		\end{aligned}
	\end{equation}
	Hence,  we can further rewrite $\widetilde{\mathcal{I}}_D(r_{0,k})$ as following
	\begin{align}\label{eqd3}
		\widetilde{\mathcal{I}}_D(r_{0,k})
		=&\sum_{j=1}^{K}\int\nolimits_{\max\{r_{0,k},H_j\}}^{R_{\mathrm{max}_j}}2\pi\bar{h}\lambda_j\frac{R_{\oplus}+H_j}{R_{\oplus}}P_jG^L_{s,j}r^{1-\alpha}G_s^U\mathrm{d}r.
	\end{align}
	Finally, by solving the above integral and by evaluating the expectation over $r_{0,k}$, the final results in Theorem \ref{coverageprobability} are derived.
	
	\section{Proof of Lemma \ref{ConditionAveragePowerPdf}}\label{ProofConditionAveragePowerPdf}
	Consider that the typical gUE is associated with the $k$-th tier, the CDF of $\xi_{0,k}$ is computed as
	\begin{align}\label{AppG1}
		F_P(\xi|\mathcal{A}_{P,k})=&\mathbb{P}\left\{\xi_{0,k}\le\xi\Big|\xi_{0,k}>\max_{{1\le j\le K, j\neq k}}\xi_{0,j}\right\}
		=\frac{\mathbb{P}\left\{\max\limits_{{1\le j\le K, j\neq k}}\xi_{0,j}<\xi_{0,k}\le\xi\right\}}{\mathbb{P}\left\{\xi_{0,k}>\max\limits_{{1\le j\le K, j\neq k}}\xi_{0,j}\right\}}.
	\end{align}
	Note that the denominator of \eqref{AppG1} is exactly the association probability presented in Lemma \ref{AssociationProbabilityPbA}, i.e.   $\mathbb{P}\left\{\xi_{0,k}>\max\limits_{\substack{1\le j\le K,~ j\neq k}}\xi_{0,j}\right\}=\mathcal{A}_{P,k}$; while the numerator can be further computed as
	\begin{align}\label{AppG2}
		&\mathbb{P}\left\{\max\limits_{\substack{1\le j\le K\\ j\neq k}}\xi_{0,j}<\xi_{0,k}\le\xi\right\}
		=\int_{\xi_{\mathrm{min}_k}}^{\xi}\prod_{\substack{j=1,j\neq k}}^{K}\mathbb{P}\left\{\xi_{0,k}>\xi_{0,j}\right\}f_P(\xi_{0,k},\lambda_k,H_k)\mathrm{d}\xi_{0,k}\nonumber\\
		=&\int_{\xi_{\mathrm{min}_k}}^{\xi}\prod_{\substack{j=1,j\neq k}}^{K}\big(1-F_P(\xi_{0,k},\lambda_j,H_j)\big)f_P(\xi_{0,k},\lambda_k,H_k)\mathrm{d}\xi_{0,k}\\
		=&\int_{\xi_{\mathrm{min}_k}}^{\xi}\prod_{j=1,j\neq k}^{K}\exp \left(-\pi\lambda_j\frac{R_{\oplus}+H_j}{R_{\oplus}}\left(\left(\frac{\bar{h}P_jG^L_{m,j}G_m^U}{\xi_{0,k}}\right)^{\frac{2}{\alpha}}-H_j^2\right)\right)f_P(\xi_{0,k},\lambda_k,H_k)\mathrm{d}\xi_{0,k}\nonumber.
	\end{align}
	Then by substituting \eqref{AppG2} into \eqref{AppG1} and by taking the partial derivative of $F_P(\xi|\mathcal{A}_{P,k})$ with respect to $\xi$, the conditional PDF of $\xi_{0,k}$ is derived.
	\section{Proof of Theorem \ref{coverageprobability2}}\label{proofcoverageprobability2}
	The proof follows a similar methodology presented in Appendix \ref{proofcoverageprobability}. In specific, based on the definition of the PbA scheme and by the law of the total probability, the coverage probability is formulated as $
	\mathcal{P}_{P}(\beta)=\sum_{k=1}^{K}\mathcal{P}_{P,k}(\beta)\mathcal{A}_{P,k}
	$, 
	where $\mathcal{P}_{P,k}(\beta)$ is given by
	\begin{equation}
		\begin{aligned}
			\mathcal{P}_{P,k}(\beta)=&
			\mathbb{P}\left\{\frac{P_kh_{0,k}G_{m,k}^LG^U_mr_{0,k}^{-\alpha}}{\sum\nolimits_{x_{i,j}\in\cup_{j=1}^{K}\Phi_j^{\backslash x_{0,k}}} P_jh_{i,j}r_{i,j}^{-\alpha}G^L_{s,j}G^U_s+\sigma^2}\ge\beta\right\}\\
			\approx&\mathbb{E}_{\xi_{0,k}}\left\{\mathbb{P}\left\{h_{0,k}\ge\frac{\beta\bar{h}\big(\widetilde{\mathcal{I}}_{P}(\xi_{0,k})+\sigma^2\big)}{\xi_{0,k}}\right\}\right\},
		\end{aligned}
	\end{equation}
	where $\widetilde{\mathcal{I}}_P(\xi_{0,k})$ is the mean of the aggregate interference signal power observed at the gUE with the employment of the PbA scheme, which is computed as
	\begin{equation}
		\begin{aligned}
			{\widetilde{\mathcal{I}}_P}(\xi_{0,k})= & \mathbb{E}\left\{\sum\nolimits_{x_{i,j}\in\cup_{j=1}^{K}\Phi_j^{\backslash x_{0,k}}} P_jh_{i,j}r_{i,j}^{-\alpha}G^L_{s,j}G^U_s\right\}
			=\bar{h}\mathbb{E}\left\{\sum\nolimits_{x_{i,j}\in\cup_{j=1}^{K}\Phi_j^{\backslash x_{0,k}}} P_jr_{i,j}^{-\alpha}G^L_{s,j}G^U_s\right\} \\
			=& \sum_{j=1}^{K}\int\nolimits_{0}^{2\pi}\int\nolimits_{\theta} \bar{h}\lambda_jP_j G^L_{s,j}r_{i,j}^{-\alpha}G^U_s(R_{\oplus}+H_j)^2\sin \theta\mathrm{d}\theta\mathrm{d}\phi\\
			=& \sum\limits_{j=1}^{K}\int_{\max\left\{\left({\varrho_k}/{\xi_{0,k}}\right)^{1/\alpha},H_j\right\}}^{R_{\mathrm{max}_j}}2\pi\bar{h}\lambda_j\frac{R_{\oplus}+H_j}{R_{\oplus}}P_jG^L_{s,j}r^{1-\alpha}G_s^U\mathrm{d}r,
		\end{aligned}
	\end{equation}
	where the last two steps directly follow from the methodology used for the derivation of $\widetilde{\mathcal{I}}_D(r_{0,k})$ presented in Appendix \ref{proofcoverageprobability}, but with a different lower limit for the integral. Finally, by applying the CDF of $h_{0,k}$, and by evaluating the average over $\xi_{0,k}$, the results in Theorem \ref{coverageprobability2} are proven.
	\section{Proof of Proposition \ref{RayleighDbA}}\label{RayleighDbAproof}
	We initially evaluate the coverage probability for DbA scheme under the Rayleigh fading. Considering that the typical gUE is associated with the $k$-th tier, then the coverage probability achieved with the DbA scheme is computed as
	\begin{align}\label{raycp}
		\mathcal{P}^{\dagger}_{D,k}(\beta)=&\mathbb{E}\left\{\mathbb{P}\left\{h_{0,k}\ge\frac{\beta\big(\sum\nolimits_{x_{i,j}\in\cup_{j=1}^{K}\Phi_j^{\backslash x_{0,k}}} P_jh_{i,j}r_{i,j}^{-\alpha}G^L_{s,j}G^U_s+\sigma^2\big)}{P_kG^L_{m,k}G^U_mr_{0,k}^{-\alpha}}\right\}\right\}\nonumber\\
		=&\mathbb{E}\left\{\exp\left(-\frac{\beta\big(\sum\nolimits_{x_{i,j}\in\cup_{j=1}^{K}\Phi_j^{\backslash x_{0,k}}} P_jh_{i,j}r_{i,j}^{-\alpha}G^L_{s,j}G^U_s+\sigma^2\big)}{P_kG^L_{m,k}G^U_mr_{0,k}^{-\alpha}}\right)\right\}\nonumber\\
		=&\mathbb{E}_{r_{0,k}}\left\{\mathbb{E}\left\{\exp\left(-{\varphi_D\sum\nolimits_{x_{i,j}\in\cup_{j=1}^{K}\Phi_j^{\backslash x_{0,k}}} P_jh_{i,j}r_{i,j}^{-\alpha}G^L_{s,j}G^U_s}\right)\right\}\exp(-\varphi_D\sigma^2)\right\}\nonumber\\
		=&\int_{H_k}^{R_{\mathrm{max}_k}}\mathcal{L}^{\dagger}_{\mathcal{I}_D}(\varphi_D)\exp(-\varphi_D\sigma^2)f_{D,k}(r_{0,k}|\mathcal{A}_{D,k})\mathrm{d}r_{0,k},
	\end{align}
	where the first step is based on the fact that for the Rayleigh fading model, the channel power gain is an exponential random variable with mean one, i.e. $h\sim \exp(1)$; $\varphi_D=\frac{\beta}{P_kG_{m,k}^LG_m^Ur_{0,k}^{-\alpha}}$ and $\mathcal{L}^{\dagger}_{\mathcal{I}_D}(z)$ is the Laplace transform of the interference, which is evaluated as following, i.e.
	\begin{equation}
		\begin{aligned}\label{lpeq}
			\mathcal{L}_{\mathcal{I}_D}^{\dagger}(\varphi_D)
			=&\mathbb{E}\left\{\prod\nolimits_{x_{i,j}\in\cup_{j=1}^{K}\Phi_j^{\backslash x_{0,k}}}\mathbb{E}_{h_{i,j}}\bigg\{\exp\left(-\varphi_D P_jh_{i,j}r_{i,j}^{-\alpha}G^L_{s,j}G^U_s\right)\bigg\}\right\}\\
			=&\mathbb{E}\left\{\prod\nolimits_{x_{i,j}\in\cup_{j=1}^{K}\Phi_j^{\backslash x_{0,k}}}\exp\left(\frac{1}{1+ \varphi_DP_jr_{i,j}^{-\alpha}G^L_{s,j}G^U_s}\right)\right\}\\
			=&\mathbb{E}\left\{\prod_{j=1}^{K}\exp\left(-2\pi\lambda_j\frac{R_{\oplus}+H_j}{R_{\oplus}}\int_{\max\{r_{0,k},H_j\}}^{R_{\mathrm{max}_j}}\left(1-\frac{1}{1+ \varphi_D P_jr^{-\alpha}G^L_{s,j}G^U_s}\right)r\mathrm{d}r\right)\right\},
		\end{aligned}
	\end{equation}
	where the last step is obtained by using the probability generating function of the PPP \cite{Haenggi}, where the above integral can be solved based on \cite[3.194.5]{intetable}. Finally, by applying the law of total probability, the coverage probability of DbA scheme in the Rayleigh fading scenario is derived. Note that the proof of the coverage probability of PbA and RbA schemes in the Rayleigh fading scenario, follows a similar methodology and thus is omitted. 
	\section{Proof of Proposition \ref{nonfadingDbA}}\label{nonfadingDbAproof}
	We provide the detailed proof of the coverage probability for the DbA scheme in the non-fading scenario, while a similar approach can be easily applied to the PbA and RbA schemes. Hence, by considering that the DbA scheme is employed and the typical gUE is associated with the $k$-th tier, the coverage probability  is formulated as
	\begin{equation}
	\begin{aligned}\label{AppF1}
		\mathcal{P}^{\ddagger}_{D,k}(\beta)=&\mathbb{P}\left\{\frac{P_kG_{m,k}^LG^U_mr_{0,k}^{-\alpha}}{\sum\nolimits_{x_{i,j}\in\cup_{j=1}^{K}\Phi_j^{\backslash x_{0,k}}} P_jr_{i,j}^{-\alpha}G^L_{s,j}G^U_s+\sigma^2}\ge\beta\right\}\\
		=&\mathbb{P}\left\{P_kG_{m,k}^LG^U_mr_{0,k}^{-\alpha}-\beta\sum\nolimits_{x_{i,j}\in\cup_{j=1}^{K}\Phi_j^{\backslash x_{0,k}}} P_jr_{i,j}^{-\alpha}G^L_{s,j}G^U_s\ge\beta\sigma^2\right\}\\
		=&\int_{H_k}^{R_{\mathrm{max}_k}}\mathbb{P}\left\{\chi(r_{0,k})\ge\beta\sigma^2\right\}f_{D,k}(r_{0,k}|\mathcal{A}_{D,k})\mathrm{d}r_{0,k},
	\end{aligned}
	\end{equation}
	where $\chi(r_{0,k})=\mathcal{S}(r_{0,k})-\beta\mathcal{I}_D(r_{0,k})$, $\mathcal{S}(r_{0,k})=P_kG_{m,k}^LG^U_mr_{0,k}^{-\alpha}$, and
	\begin{equation}
		\mathcal{I}_D(r_{0,k})=\sum\nolimits_{x_{i,j}\in\cup_{j=1}^{K}\Phi_j^{\backslash x_{0,k}}} P_jr_{i,j}^{-\alpha}G^L_{s,j}G^U_s.
	\end{equation}
	Then, by applying the Gil-Pelaez inversion theorem \cite{PELAEZ,Sayehvand,Guo}, the probability term in \eqref{AppF1} can be further computed as
	\begin{align}\label{AppF2}
		\mathbb{P}\left\{\chi(r_{0,k})\ge\beta\sigma^2\right\}=\frac{1}{2}-\frac{1}{\pi}\int_{0}^{\infty}\frac{\Im\left\{\psi_{\chi(r_{0,k})}(t)\exp(it\beta\sigma^2)\right\}}{t}\mathrm{d}t,
	\end{align}
	where $\Im\{ \cdot \}$ is the imaginary operator and $\psi_{\chi(r_{0,k})}$ is the characteristic function of $\chi(r_{0,k})$ which is given by
	\begin{align}\label{AppF3}
		\psi_{\chi(r_{0,k})}(t)=\mathbb{E}\left\{\exp\bigg(-it\big(\mathcal{S}(r_{0,k})-\beta\mathcal{I}_D(r_{0,k})\big)\bigg)\right\}=e^{-it\mathcal{S}(r_{0,k})}\mathcal{L}^{\ddagger}_{\mathcal{I}_D}(-it\beta),
	\end{align}
	where $\mathcal{L}^{\ddagger}_{\mathcal{I}_D}(z)=\mathbb{E}\left\{e^{-z\mathcal{I}_{D}(r_{0,k})}\right\}$ is the Laplace transform of the interference, which can be evaluated as following
	\begin{equation}
	\begin{aligned}\label{AppF4}
		\mathcal{L}^{\ddagger}_{\mathcal{I}_D}(z)&=\mathbb{E}\left\{\prod\nolimits_{x_{i,j}\in\cup_{j=1}^{K}\Phi_j^{\backslash x_{0,k}}}\exp\left(-{z P_jr_{i,j}^{-\alpha}G^L_{s,j}G^U_s}\right)\right\}\\
	=&\mathbb{E}\left\{\prod_{j=1}^{K}\exp\left(-2\pi\lambda_j\frac{R_{\oplus}+H_j}{R_{\oplus}}\int_{\max\{r_{0,k},H_j\}}^{R_{\mathrm{max}_j}}\big(1-\exp\left(-{z P_jr^{-\alpha}G^L_{s,j}G^U_s}\right)\big)r\mathrm{d}r\right)\right\}.
	\end{aligned}
	\end{equation}
	Hence, by solving the above integral and by substituting \eqref{AppF2}-\eqref{AppF4} into $\eqref{AppF1}$, the final expression of $\mathcal{P}_{D,k}^{\ddagger}(\beta)$ is obtained. Finally, by the law of total probability, the coverage probability of the DbA scheme for the non-fading scenario is derived. 
	\bibliographystyle{IEEEtran}
	\bibliography{refs}
\end{document}